\documentclass[twocolumn,showpacs,laps,pra,amsmath,amssymb]{revtex4}

\usepackage{bm,color}
\usepackage{graphicx}
\usepackage{epstopdf}
\usepackage{amsthm}

\usepackage[normalem]{ulem}
\usepackage[usenames,dvipsnames]{xcolor}

\newcommand{\beq}{\begin{equation}}
\newcommand{\eeq}{\end{equation}}
\newcommand{\nn}{\nonumber \\}
\newcommand{\bra}[1]{ \langle{#1} |}
\newcommand{\ket}[1]{ |{#1} \rangle}
\newcommand{\braket}[2]{\langle {#1} | {#2} \rangle}

\newcommand{\ramse}{RAMSE}
\newcommand{\amse}{AMSE}
\newcommand{\mse}{MSE}
\newcommand{\rmse}{RMSE}
\newcommand{\unb}{$U(1)$-unbiased}

\newtheorem{theorem}{Theorem}
\newtheorem{lemma}{Lemma}

\begin{document}
\title{Optimal Heisenberg-style bounds for the average performance of arbitrary phase estimates}
\author{Dominic W. Berry${}^1$, Michael J. W. Hall${}^2$, Marcin Zwierz${}^2$ and Howard M. Wiseman${}^2$}
\affiliation{${}^1$Department of Physics and Astronomy, Macquarie University, Sydney, NSW 2109, Australia\\
${}^2$Centre for Quantum Computation and Communication Technology (Australian Research Council), Centre for Quantum Dynamics, Griffith University, Brisbane, QLD 4111, Australia}
\date{\today}

\begin{abstract}
The ultimate bound to the accuracy of phase estimates is often assumed to be given by the Heisenberg limit.
Recent work seemed to indicate that this bound can be violated, yielding measurements with much higher accuracy than was previously expected.
The Heisenberg limit can be restored as a rigorous bound to the accuracy provided one considers the accuracy averaged over the possible values of the unknown phase, as we have recently shown [Phys.~Rev.~A \textbf{85}, 041802(R) (2012)].
Here we present an expanded proof of this result together with a number of additional results, including the proof of a previously conjectured stronger bound in the asymptotic limit.
Other measures of the accuracy are examined, as well as other restrictions on the generator of the phase shifts.
We provide expanded numerical results for the minimum error and asymptotic expansions.
The significance of the results claiming violation of the Heisenberg limit is assessed, followed by a detailed discussion of the limitations of the Cram\'er-Rao bound.
\end{abstract}

\maketitle

\section{Introduction}
Phase  estimation  is the basis for much precision measurement.
Optical interferometers offer highly accurate  measurements of length, and atomic phase measurements  provide  highly accurate  measurements of time, as well as other physical quantities like magnetic field \cite{caves81,phasesens,phasecomm,phaseres,WisMil10}.
In optics, most measurements are limited by the shot-noise limit, where
the accuracy  scales as $1/\sqrt{\langle N \rangle}$, where $N$ is the photon number operator.
In contrast, it is normally assumed that the fundamental limit is the Heisenberg limit, where the accuracy scales as $1/\langle N \rangle$ \cite{heislim,zwierz,GLM06}.
This potentially provides far greater accuracy, but is extremely difficult to achieve in practice because it requires highly nonclassical states of light, as well as arbitrarily high efficiencies \cite{dorner,kacprowicz,escher}.
Any amount of loss will cause the scaling to revert to $1/\sqrt{\langle N \rangle}$ for large $\langle N \rangle$ \cite{escher}.

Recently a number of papers suggested that the Heisenberg limit is not the fundamental limit to accuracy, and that a better scaling constant or even a higher power of $\langle N \rangle$ might be possible.
In Ref.~\cite{ani}, Anisimov \textit{et al.}\ gave a proposal for violating the Heisenberg limit by a small amount.
In another work, Zhang \textit{et al.}\ \cite{zhang}  proposed  a scheme offering zero phase uncertainty with finite $\langle N \rangle$. 
Finally, in Ref.~\cite{rivas} Rivas and Luis presented a proposal for obtaining scaling as $1/\langle N \rangle^p$ for $p>1$.
A qualitatively different proposal for violating the Heisenberg limit is that based on nonlinear interferometry \cite{nonlinear,napolit}.
However, that work differs in its use of terminology; it does not violate the Heisenberg limit in the sense we use here  (see Sec.~\ref{sec:nonlin}) \cite{zwierz}.

A common feature of proposals to violate the Heisenberg limit is that they only work for a limited range of phases.
Additional phase information would be needed to confine the phase to within the region where the measurement is accurate.
One can consider first using a sequence of measurements to ensure that the phase lies within a suitable region, 
then using the super-Heisenberg measurement. 
If the \emph{overall} measurement (consisting of the sequence of individual measurements) could yield better accuracy than the Heisenberg limit, then it could be regarded as providing a true improvement.
On the other hand, if the resources required to localise the phase to the required region result in an overall measurement with accuracy that is not better than the Heisenberg limit, then the accuracy of the super-Heisenberg measurements would seem to be illusory.

An analogous situation was seen in considering the reciprocal-peak-likelihood as a measure of uncertainty.
In Ref.~\cite{shapiro,shapiro2} a technique was proposed that would apparently yield super-Heisenberg accuracy in terms of reciprocal-peak-likelihood.
Later work found that, in practice, the proposal resulted in accuracy that was worse than the Heisenberg limit \cite{caves}.
Another example is that of NOON states.
NOON states yield phase information scaling as the Heisenberg limit, but require initial phase information with similar accuracy.
In that case, it is known how to combine measurements from multiple states to obtain an overall measurement that scales at the Heisenberg limit \cite{multi,berry}.

To evaluate whether the super-Heisenberg measurements would be able to yield an overall measurement violating the Heisenberg limit, we examined the case where the mean-square error is averaged over all phase shifts.
We showed that the Heisenberg limit provides a rigorous lower bound to the square root of the average mean-square error ({\ramse}) in such a case \cite{rapid}.
Therefore, no scheme that apparently beats the Heisenberg limit for a small range of phase could be used to construct an overall measurement starting from an unknown phase that beats the Heisenberg limit.
 An alternative approach is to determine the  bound if the initial phase is restricted to a given range.
In Ref.~\cite{Hall12} it was shown that  with such a restriction  the usual Heisenberg limit can be multiplied by a factor proportional to the phase range, and further results have been given in Refs.~\cite{tsang,gm,nair,hwprx}.
An alternative approach has yielded a bound on the average of the error at just two locations \cite{glm}.

The specific result from Ref.~\cite{rapid} is
\begin{equation}
\label{eq:heis0}
\delta\hat{\Phi} \ge \frac k{\langle G +1 \rangle},
\end{equation}
where $\delta\hat{\Phi}$ is the {\ramse}, $G$ is the generator of the phase shifts, which is here assumed to have nonnegative integer eigenvalues, and $k$ is a constant.
These quantities are explained in Sec.~\ref{sec:merit} below.
We have analytically proven that this inequality holds with $ k = k_A:=\sqrt{2\pi/e^3}\approx 0.5593$ \cite{rapid}.
In Sec.~\ref{sec:proof} we give the full proof of that result, as well as a generalised result in terms of the absolute value of $G$ in the case where $G$ also has negative integer eigenvalues.

Numerical calculations indicate that the inequality holds with the larger scaling constant $k = k_C \approx 1.3761$.
We give the detailed numerical results in Sec.~\ref{sec:numer}, indicating that this result holds both for the {\ramse} and the error estimated using the Holevo variance.
In Sec.~\ref{sec:asymp} we calculate asymptotic expansions for the {\ramse}, providing strong analytic support for the scaling constant $k_C$, and proving that $k=k_C$ is valid in the asymptotic limit $\langle G\rangle\rightarrow\infty$.
We examine the scaling with the number of probe states in Sec.~\ref{sec:probe}, then give a detailed discussion of the papers claiming violation of the Heisenberg limit in Sec.~\ref{sec:viol}.
The Cram\'er-Rao bound and the error propagation formula are commonly used in examining the Heisenberg limit, but have some limitations; these are discussed in Sec.~\ref{sec:cr}. 

\section{Figures of merit for average phase resolution}
\label{sec:merit}
There are a number of different figures of merit for  phase measurements. 
Before describing these, we first introduce some notation, 
largely following Ref.~\cite{rapid}.
The random variable for the phase shift of the system is $\Phi$,
and the random variable for the estimate of that phase shift is $\hat{\Phi}$.
The error in the phase estimate is $\Theta = \hat{\Phi}-\Phi$.
We use capital letters for the random variables;
the corresponding values and measurement outcomes are denoted by the corresponding lower case letters ($\phi$, $\hat\phi$, and $\theta$).

We consider a Hilbert space with a phase shift operator $G$.
In the completely general case, the only restriction is that the eigenvalues of $G$ must be integers.
We may also consider the specific case where the eigenvalues are all nonnegative integers, in which case we denote the operator by $N$.
This includes, for example, the case of photon number.
Alternatively, if the eigenvalues include all integers, such as for angular momentum, we use the symbol $J$.

The phase shift is described by the unitary operator $\exp(-iG\Phi)$.
That is, the probe state $\rho_0$ becomes $\rho_\phi:=e^{-iG\phi}\rho_0 e^{iG\phi}$.
The detection method used to estimate $\phi$ is described by a  positive-operator valued measure (POVM) $\{M_{\hat\phi}\}$.
Hence, the probability distribution is given by $p(\hat\phi|\phi)={\rm Tr}(M_{\hat\phi}\rho_\phi )$.
Because phase is only defined modulo $2\pi$, we do not distinguish between $\phi$ and $\phi+2\pi$, or between $\hat\phi$ and $\hat\phi+2\pi$.
This means that $p(\hat\phi|\phi)=p(\hat\phi+2\pi k|\phi)$ for any integer $k$, and $p(\hat\phi|\phi)$ is normalised over  a  (arbitrary) $2\pi$ interval.

\subsection{Root-mean-square error}
\label{sec:rms}
The most common figure of merit for a measurement is the square root of the mean-square error ({\mse}).
We will call this the {\rmse}.
For a specific phase shift, $\phi$, the {\mse} is given by
\begin{equation}
\label{eq:oldv}
 (\Delta_\phi^{\phi_r} \hat{\Phi})^2  := \int_{\phi_r-\pi}^{\phi_r+\pi} d\hat\phi \, (\hat\phi-\phi)^2 p(\hat\phi|\phi).
\end{equation}
There is a subtlety in that for phase, values that differ by $2\pi$ are equivalent, which means that a range of $2\pi$ must be specified for the integral.
However, the reference phase shift, $\phi_r$, is arbitrary, and the value that is obtained for the {\mse} will depend on $\phi_r$.
Ideally $\phi$ should be near the centre of the range.
If it is near one of the bounds of the range then the {\mse} will be unreasonably large.

To solve this problem, it is convenient to take the difference  $\hat\phi-\phi$ modulo $(-\pi,\pi]$. 
That is, we add or subtract a multiple of $2\pi$ such that the value obtained is in the range $(-\pi,\pi]$.
It is important to note that this convention can only decrease the value obtained for the {\mse}.
In this work we are concerned with placing lower bounds on the {\mse}.
We prove these lower bounds for the {\mse} with the difference defined modulo $2\pi$.
Because this {\mse} is no larger than that obtained without taking the difference modulo  $(-\pi,\pi]$, all results hold for that case as well.
 Thus, it is natural to work with the minimum {\mse},  given by
\begin{align}
\label{eq:newv}
(\Delta_\phi \hat{\Phi})^2 &:= \int_{\phi_r-\pi}^{\phi_r+\pi} d\hat\phi \,  \,  \left\{(\hat\phi-\phi) \!\! \mod \!(-\pi,\pi]\right\}^2  p(\hat\phi|\phi) \nn
&= \int_{\phi-\pi}^{\phi+\pi} d\hat\phi \, (\hat\phi-\phi)^2 p(\hat\phi|\phi) ,
\end{align}
where we have used the fact that $p(\hat\phi|\phi)$ repeats modulo $2\pi$. 
 It follows that
\beq
\Delta_\phi\hat{\Phi} \equiv  \Delta_\phi^{\phi} \hat{\Phi}\le\Delta_\phi^{\phi_r} \hat{\Phi}
\eeq
for any reference phase $\phi_r$.

The above $\Delta_\phi \hat{\Phi}$ is a measure of the accuracy of the phase measurement only for a specific phase shift $\phi$.
It is trivial to see that one can always choose a measurement such that the {\mse} can be zero for a specific phase shift, $\phi_0$: 
the trivial measurement that always yields the result $\hat\phi=\phi_0$.
In reality, for a phase measurement the phase shift is unknown; otherwise a measurement would be unnecessary. 
To be useful, a measurement must give accurate results for a range of phase shifts.

A rigorous way of taking account of the range of phase is to average the figure of merit over the phase shift.
For the {\mse} one would use
\begin{equation}
\int_{-\pi}^{\pi} d\phi \, p(\phi) (\Delta_\phi \hat{\Phi})^2, 
\end{equation}
where $p(\phi)$ is a probability distribution describing the prior information about the phase shift.
In this work we consider the case that there is no prior information, so $p(\phi)=1/2\pi$.
Then the average {\mse} ({\amse}) is given by
\begin{equation}
(\delta\hat\Phi)^2 := \frac 1{2\pi}\int_{-\pi}^{\pi} d\phi \, (\Delta_\phi \hat{\Phi})^2.
\end{equation}
One then finds that 
\begin{align}
\label{eq:aveq}
(\delta\hat\Phi)^2 =\int_{-\pi}^{\pi} d\theta \, \theta^2 \,  \bar p(\theta) = \langle \Theta^2 \rangle.
\end{align} 
Here $\bar p(\theta)$ is the probability density for the error in the phase estimate $\Theta=\hat{\Phi}-\Phi$, and is defined by
\begin{equation}
\bar p(\theta) :=  \frac 1{2\pi}\int_{-\pi}^{\pi} d\phi \,  p(\theta+\phi|\phi).
\end{equation}
We call $\delta\hat\Phi$ the {\ramse}, because it is averaged over $\phi$ before taking the square root, whereas the {\rmse} $\Delta_\phi\hat\Phi$ is for a specific $\phi$.

Equation \eqref{eq:aveq} holds because the mean-square error is a linear figure of merit. 
A general figure of merit for the accuracy of a phase estimate $\hat\Phi$ can be defined as a 
functional, $F$, that takes as input a probability density in $\hat\phi$, and outputs a scalar. 
 In the case that $F$ is linear, we find that 
\begin{align}
\label{eq:aveq2}
\int_{-\pi}^{\pi} d\phi \, p(\phi) F(p(\hat\Phi|\phi)) 
= F(\bar p(\Theta)).
\end{align}
 This means that, for linear measures, the average figure of merit and the figure of merit of the average distribution are equivalent.

More generally, consider a convex figure of merit; that is, one that satisfies
\begin{equation}
F(t p_1(\hat\Phi) + (1-t) p_2(\hat\Phi)) \le t F(p_1(\hat\Phi)) + (1-t) F(p_2(\hat\Phi)),
\end{equation}
 for $t\in[0,1]$. 
By using Jensen's inequality, one obtains
\begin{align} \label{eq:gfomap}
\int_{-\pi}^{\pi} d\phi \, p(\phi) F(p(\hat\Phi|\phi)) &\ge F\left( \int_{-\pi}^{\pi} d\phi \, p(\phi) p(\hat\Phi|\phi) \right) \nn
&= F(\bar p(\Theta)).
\end{align}

What this means is that, if the figure of merit is convex, then placing a lower bound on the figure of merit for the average distribution also
provides a lower bound on the average of the figure of merit.
That is the approach we use in this work; we find lower bounds on the figure of merit for the average distribution, which also hold for the average of the figure of merit.

\subsection{Holevo variance and average bias}
\label{sec:hol}
There are alternative measures of the spread which are similar to the {\mse} 
but which are specifically defined for phase. These are typically defined in terms of 
the average of the exponential of the phase, $\langle e^{i\hat\Phi}\rangle$.
In the case that the phase distribution is sharply peaked, then this quantity will be close to 1.
One possibility for quantifying the uncertainty in the phase is $2(1-|\langle e^{i\hat\Phi}\rangle|)$ \cite{Collett}; 
another is $1-|\langle e^{i\hat\Phi}\rangle|^2$ \cite{Opatrny,Forbes}.

A measure of this type with some nice properties is that  proposed by Holevo \cite{holcov}, 
\begin{equation}
V_{H,\phi}(\hat\Phi) := |\langle e^{i\hat\Phi}\rangle_{\phi}|^{-2}-1,
\end{equation}
which has been dubbed the Holevo variance \cite{WisKil97}. 
Here the subscript $\phi$ indicates that the variance is determined for a specific value of the phase shift.
That is,
\begin{equation}
V_{H,\phi}(\hat\Phi) := \left|\left\langle \int_{-\pi}^{\pi} d\hat\phi \, e^{i\hat\phi}p(\hat\phi|\phi)\right\rangle \right|^{-2}-1.
\end{equation}
In this case there is no ambiguity in choosing the bounds of the integral, because the argument is clearly periodic modulo $2\pi$. 

A minor problem with this definition is that it does not penalise biased estimates.
However, this is easily corrected by using the modified definition
\begin{equation}
{\rm Var}_{H,\phi}\hat\Phi := {\rm Re}\langle e^{i(\hat\Phi-\phi)}\rangle_{\phi}^{-2}-1.
\end{equation}
If the measurement is ``\unb'', in the sense that
\begin{equation}
\phi = \arg [\langle e^{i\hat\Phi}\rangle_{\phi}],
\end{equation}
then these two expressions for the Holevo variance are equivalent.

The Holevo variance is a convex functional of the probability distribution.
From  Eq.~\eqref{eq:gfomap},  this means that one can place a lower bound on the average Holevo variance by considering the Holevo variance of the average distribution.
That is,
\begin{equation}
(\delta_H \hat\Phi)^2 := ({\rm Re}\langle e^{i\Theta} \rangle)^{-2}-1
\end{equation}
is a lower bound on the average value of ${\rm Var}_{H,\phi}(\hat\Phi)$.
In this paper we do not discuss the Holevo variance without averaging over $\phi$, so we will refer to $(\delta_H \hat\Phi)^2$ as the Holevo variance.

In the case that the average  distribution $\overline{p}(\theta)$ is  {\unb}, in the sense that $\langle e^{i\Theta} \rangle$ is real and positive,  then 
\begin{equation}
(\delta_H \hat\Phi)^2 = |\langle e^{i\Theta} \rangle|^{-2}-1.
\end{equation}
If the average  distribution  is biased, then it can be modified to obtain a {\unb} measurement.
Taking
$\theta_{\rm av} := \arg\langle e^{i\Theta} \rangle$,
we can replace measurement operators $M_{\hat\Phi}$ with
\begin{equation}
M'_{\hat\Phi} = M_{\hat\Phi+\theta_{\rm av}}.
\end{equation}
Then, for these new measurement operators, $p'(\hat\phi|\phi)=p(\hat\phi+\theta_{\rm av}|\phi)$, so
\begin{align}
\langle e^{i\Theta} \rangle_{M'} &= \frac 1{2\pi} \int_{-\pi}^{\pi} d\hat\phi \int_{-\pi}^{\pi} d\phi  \, e^{i(\hat\phi-\phi)}p(\hat\phi+\theta_{\rm av}|\phi) \nn
&= \frac 1{2\pi} \int_{-\pi}^{\pi} d\hat\phi \int_{-\pi}^{\pi} d\phi  \, e^{i(\hat\phi-\theta_{\rm av}-\phi)}p(\hat\phi|\phi) \nn
&= e^{-i\theta_{\rm av}}\langle e^{i\Theta} \rangle_{M}.
\end{align}
Hence this modification of the measurement yields a {\unb} average measurement.

With this condition, we can bound the mean-square error by using the following inequality, 
\begin{equation} \label{earlier}
|\langle e^{i\Theta} \rangle| = \langle \cos\Theta \rangle \ge \cos\sqrt{\langle \Theta^2\rangle}, 
\end{equation}
where we have used the fact that $\cos\sqrt{x}$ is a convex function, along with Jensen's inequality.
There are alternative ways to bound the mean-square error, but this particular inequality will be useful 
in  Appendix~\ref{sec:tighter}. It also has the nice property that it can be saturated, 
for a probability distribution that is just delta functions at $\pm\sqrt{\langle \Theta^2\rangle}$.
Now consider the limit where the mean-square error $(\delta\hat\Phi)^2=\langle \Theta^2\rangle$ is small.
Expanding as a Maclaurin series in this small parameter, we obtain
\begin{align}
\label{eq:vaholin}
(\delta\hat\Phi)^2 &\ge \left( \arccos \{[(\delta_H \hat\Phi)^2+1]^{-1/2}\} \right)^2 \nn
&= (\delta_H \hat\Phi)^2 - \frac 23 (\delta_H \hat\Phi)^4 + O((\delta_H \hat\Phi)^6).
\end{align}
This means that, except for higher-order terms, the Holevo variance lower bounds the {\amse} from below, 
so asymptotically we have $(\delta_H \hat\Phi)^2 \lesssim (\delta\hat\Phi)^2$. 

 We can also use the Holevo variance to bound the {\amse} from above.
Using  the fact that $\cos\theta\le 1-2\theta^2/\pi^2$ on the interval $[-\pi,\pi]$, we have the inequality
\begin{equation} \label{ineq2}
\langle \cos\Theta \rangle \le 1-\frac{2\langle \Theta^2 \rangle}{\pi^2}.
\end{equation}
Using this, we have  
\begin{align}
(\delta\hat\Phi)^2 &\le \frac{\pi^2}2 \left( 1-[(\delta_H \hat\Phi)^2+1]^{-1/2}\right) \nn
&= \frac{\pi^2}4 (\delta_H \hat\Phi)^2 - \frac{3\pi^2}{16} (\delta_H \hat\Phi)^4 + O((\delta_H \hat\Phi)^6).
\end{align}
 The reason for the factor of $\pi^2/4$ is that even for small variance, the main contribution to the {\amse} 
can be from large phase errors.
The inequality (\ref{ineq2}) is saturated for a distribution that has contributions at $\pm \pi$.  

Returning to the asymptotic lower bound (\ref{eq:vaholin}), its significance is that 
any lower bound on the Holevo variance is also asymptotically a lower bound on the {\amse}. 
In particular, it is known that for canonical phase measurements on a single-mode field there is the tight asymptotic lower bound on the Holevo variance
$\delta_H\hat{\Phi}\gtrsim k_C /\langle N\rangle$ with $k_C  :=2(-z_A/3)^{3/2}\approx 1.3761$ (where $z_A$ is the first zero of the Airy function) \cite{berrythesis,bandilla}.
This is tight in the sense that, asymptotically, the Holevo variance is equal to this value with any difference being of higher order.
Because the Holevo variance is asymptotically a lower bound on the usual {\amse}, we must also asymptotically have the lower bound
$\delta\hat{\Phi}\gtrsim k_C /\langle N\rangle$ for canonical phase measurements on a single-mode field.
It will be shown in Sec.~\ref{sec:proof} that this is in fact a tight lower bound.

\subsection{Entropic length}

Another measure of concentration is the \emph{entropic length} \cite{entvol,Hall00}.
This is given by
\begin{equation} \label{entlength}
L(\hat\Phi) := e^{H(\Theta)},
\end{equation}
where $H(\Theta)$ is the entropy of the error probability density,
\begin{equation}
H(\Theta)= -\int_{-\pi}^{\pi} \bar p(\theta) \ln(\bar p(\theta)) d\theta.
\end{equation}
The entropy takes its largest positive value for a flat distribution, and takes large negative values as the distribution provides more information about the phase.
The negative of the entropy provides a measure of how much information about the phase is available.
The entropic length is correspondingly small for a distribution providing a lot of information about the phase.

Similar to the {\amse} or the Holevo variance, the entropic length will be small for a sharply peaked distribution.
However, in contrast to those measures, the entropic length will also be small if there are multiple sharp peaks, 
with a value roughly equal to the total width of those peaks.
The entropic length satisfies several basic properties expected for a length, discussed in Ref.~\cite{entvol}.
It can also be used to provide a lower bound to the {\ramse} via the relation \cite{Hall00}
\begin{equation} \label{entvar}
\delta \hat{\Phi} \geq (2\pi e)^{-1/2} L(\hat{\Phi}).
\end{equation}
This is because, if one were considering a distribution on the infinite line, the entropy is maximised for fixed $\hat{\Phi}$ by a Gaussian distribution, in which case $\delta \hat{\Phi} = (2\pi e)^{-1/2} L(\hat{\Phi})$.
For the case of phase, we are limited to the interval $[-\pi,\pi]$, which means that the Gaussian distribution cannot be obtained exactly.
Therefore the inequality still holds, but cannot be saturated except asymptotically. 

 In contrast to the other measures considered here, the entropy is \emph{not} convex.
This means that one needs to be cautious when considering the average entropy.
The entropy of the average distribution does not provide a lower bound on the average of the entropies.
We do not determine the lower bound on the average of entropies; this is an open problem. 

 \section{Obtaining universal bounds from nondegenerate bounds} 
\label{sec:proof}
We now present the universal form of the Heisenberg limit, which was first derived in Ref.~\cite{rapid}.
In subsection A we present the theorem showing that bounds which hold for canonical measurements on  nondegenerate systems  also hold for completely arbitrary measurements on general systems.
 In optics a single-mode field is nondegenerate, whereas the general case includes multimode interferometry. 
In subsection B we use this to provide our universal form of the Heisenberg limit. 
In Sec.~IV we will present numerical results indicating that a better scaling constant is possible.

\subsection{Mapping the general problem to a  nondegenerate  problem}

As discussed at the start of Sec.~\ref{sec:merit}, the detection method may be described by a  POVM  $\{M_{\hat\phi}\}$, which gives the probability distribution via
\begin{equation}
p(\hat\phi|\phi) = {\rm Tr}(M_{\hat\phi}\rho_\phi).
\end{equation}
A particularly useful form of  POVM  is a \emph{covariant}  POVM.
Whereas for an arbitrary  POVM  the individual $M_{\hat\phi}$ can be chosen independently of each other (except for the normalisation requirement), for a covariant  POVM  only one measurement operator may be chosen, then all others are related via the generator of shifts.
In particular,
\begin{equation}
\overline{M}_{\hat\phi} =  e^{-iG\hat\phi} \overline{M}_0 e^{iG\hat\phi}.
\end{equation}
For a covariant  POVM, the probability distribution for the error in the estimate is independent of the phase shift.
This may be shown via
\begin{align}
p(\theta+\phi| \phi) &= {\rm Tr}(\overline{M}_{\theta+\phi}\rho_\phi) \nn
&= {\rm Tr}(e^{-iG(\theta+\phi)} \overline{M}_{0}e^{iG(\theta+\phi)}e^{-iG\phi}\rho_0 e^{iG\phi}) \nn
&= {\rm Tr}(e^{-iG\theta} \overline{M}_{0}e^{iG\theta}\rho_0 ) .
\end{align}

A particular form of covariant  POVM  is the \emph{canonical}  POVM.
This can be defined as $\{ e^{-iG\phi}C_0e^{iG\phi}\}$, with \cite{Hall08}
\begin{equation}
\label{eq:can}
C_0 = \frac{1}{2\pi} \sum_d \sum_{n,n'\in S;d\leq D(n),D(n') } |n,d\rangle \langle n',d|.
\end{equation}
Here we have labelled the states with $n$ and $n'$ indicating the eigenvalues of $G$, and $d$ the degeneracy.
The function $D(n)$ gives the degeneracy for eigenvalue $n$.
$S$ denotes the spectrum of eigenvalues of $G$, which we have assumed to be the integers or a  subset thereof.
This definition of a canonical POVM is not unique in general, because it depends on the labelling of the degenerate states; a fact which was not noted in Refs.~\cite{Hall08,rapid},
and which does not affect the results therein.
 However, we will only require the simpler case of no degeneracies in what follows, where for this case the POVM is uniquely given by
$\{ e^{-iG^{(s)}\phi}C_0^{(s)} e^{iG^{(s)}\phi}\}$, with
\begin{equation}
C_0^{(s)} = \frac{1}{2\pi}\sum_{n,n'\in S} |n\rangle \langle n'|.
\end{equation}
We use $G^{(s)}$ to denote a generator with
the same spectrum of eigenvalues as $G$, but nondegenerate.

We now show  
that any average phase distribution, $\bar p(\theta)$, can be obtained by a covariant measurement, 
and that the covariant measurement result can be obtained by a canonical measurement on a system without degeneracy.
In Ref.~\cite{rapid} we obtained this result by a three-step  process: first that any 
average phase distribution, $\bar p(\theta)$, can be obtained by a covariant measurement; second that 
the covariant measurement result can be obtained by a canonical measurement; and third that 
the canonical measurement result can be obtained by a canonical measurement on a system without degeneracy.
Here we simplify the proof by combining  the second two steps.

To express these results it is convenient to modify the notation slightly.
We will use subscripts on the probability $p$ to indicate the  POVM  used.
In addition, we will indicate the state used in the probability.
In the case of the probability for the measurement error $\theta$ for the covariant  POVM, we omit $\phi$, because the probability is independent of $\phi$ as discussed above.
Therefore, we replace $p(\theta+\phi|\phi)$ with $p_{\overline{M}}(\theta|\rho_0)$.

Expressed in terms of this notation, the first result is as follows.

\begin{lemma}
\label{lem1}
For any  POVM  $\{M_{\hat\phi}\}$, there exists a covariant  POVM  $\{\overline{M}_{\hat\phi}\}$ such that for all states $\rho_0$,
\begin{equation} \label{lemrel}
p_{\overline{M}}(\theta|\rho_0) = \bar p_M(\theta|\rho_0).
\end{equation}
\end{lemma}

\begin{proof}
This result is well known \cite{holcov}, but we provide a proof here for completeness.
Given  POVM  $\{M_{\hat\phi}\}$, we define the covariant  POVM  via
\begin{equation}
\overline{M}_0 := \frac{1}{2\pi} \int_{-\pi}^{\pi} d\phi \, e^{iG\phi}M_{\phi}e^{-iG\phi}. 
\end{equation}
Then we find
\begin{align}
p_{\overline{M}}(\theta|\rho_0) &= {\rm Tr}(e^{-iG\theta} \overline{M}_{0}e^{iG\theta}\rho_0 ) \nn
&= \frac{1}{2\pi} \int_{-\pi}^{\pi} d\phi \, {\rm Tr}(e^{-iG\theta} e^{iG\phi}M_{\phi}e^{-iG\phi} e^{iG\theta}\rho_0 ) \nn
&= \frac{1}{2\pi}\int_{-\pi}^{\pi} d\phi \, p_{M} (\phi | \phi-\theta) \nn
&= \frac{1}{2\pi}\int_{-\pi}^{\pi} d\phi \, p_{M} (\phi+\theta | \phi) = \bar p_M(\theta|\rho_0).
\end{align}
In the last line we have shifted the variable of integration.
This shows the relation \eqref{lemrel} required.
\end{proof}

The second result, which is a combination of the two steps given in Ref.~\cite{rapid}, is as follows. 

\begin{lemma} \label{lemz}
Given any covariant  POVM  $\{\overline{M}_{\hat\phi}\}$ and state $\rho_0$, there exists a state without degeneracies
$\rho_0^{(s)}$ such that the probability distribution of $G^{(s)}$ for $\rho_0^{(s)}$ is the same as that of $G$ for $\rho_0$, and
\begin{equation}\label{lem2rel}
p_{C^{(s)}}(\theta|\rho_0^{(s)}) = p_{\overline{M}}(\theta|\rho_0).
\end{equation}
\end{lemma}

\begin{proof}
Choose the state without degeneracies via
\begin{align}
\rho_0^{(s)} &= 2\pi  \sum_{n,n'\in S}\ket{n'}\bra{n} \sum_{d\leq D(n),d'\leq D(n')} \langle n',d'|\rho_0|n,d\rangle \nn
& \quad \times \langle n,d|\overline{M}_0|n',d'\rangle .
\end{align}
Note that $\rho_0^{(s)}$ is indeed a density operator,  and yields the same distribution for $G^{(s)}$ as $\rho_0$ does for $G$.
The details for how to prove these facts are given in Appendix \ref{sec:lem2det}.

It can be shown that the average phase distributions  are also the same, via
\begin{widetext}
\begin{align}
p_{C^{(s)}}(\theta|\rho_0^{(s)}) &= {\rm Tr} (e^{-iG^{(s)}\theta}C_0^{(s)} e^{iG^{(s)}\theta}\rho_0^{(s)}) \nn
&=  \frac{1}{2\pi}{\rm Tr} \left(e^{-iG^{(s)}\theta}\sum_{m,m'\in S} |m\rangle \langle m'| e^{iG^{(s)}\theta} 2\pi  \sum_{n,n'\in S}\ket{n'}\bra{n} \sum_{d\leq D(n),d'\leq D(n')} \langle n',d'|\rho_0|n,d\rangle \times \langle n,d|\overline{M}_0|n',d'\rangle\right) \nn
&= \sum_{n,n'\in S} e^{i(n'-n)\theta} \sum_{d\leq D(n),d'\leq D(n')} \langle n',d'|\rho_0|n,d\rangle \langle n,d|\overline{M}_0|n',d'\rangle \nn
&= \sum_{n,n'\in S}  \sum_{d\leq D(n),d'\leq D(n')} \langle n',d'|\rho_0|n,d\rangle \langle n,d|e^{-iG\theta}\overline{M}_0e^{iG\theta}|n',d'\rangle\nn
&= \sum_{n'\in S}  \sum_{d'\leq D(n')} \langle n',d'|\rho_0 e^{-iG\theta}\overline{M}_0e^{iG\theta}|n',d'\rangle
= {\rm Tr}(e^{-iG\theta} \overline{M}_{0}e^{iG\theta}\rho_0 ) = p_{\overline{M}}(\theta|\rho_0).
\end{align}
\end{widetext}
This shows the relation \eqref{lem2rel} required.
\end{proof}

 Using these lemmas then enables us to prove our theorem that the average distribution can always be obtained by a canonical measurement on a system without degeneracies. 

\begin{theorem} \label{mainthm}
Any bound on the concentration of the canonical phase distribution of a nondegenerate system with shift generator $G^{(s)}$, under some constraint ${\cal C}$ on the distribution of $G^{(s)}$, is also a bound on the concentration of the average phase distribution $\overline{p}(\theta)$ of an arbitrary phase estimate for any shift generator $G$ having the same eigenvalue spectrum as $G^{(s)}$, providing that the probe state satisfies the same constraint ${\cal C}$ with respect to the distribution of $G$.
\end{theorem}

\noindent
\textit{Remarks.} A measure of the concentration is a functional of the probability distribution, and includes the mean-square error, the Holevo variance, and the entropic length.
 For measures that are convex, such as the mean-square error and Holevo variance, lower bounds on the measure for the average distribution provide lower bounds on the average of that measure (see Sec.~\ref{sec:rms}). 
By the distribution of $G$, we mean the probability distribution for the eigenvalues of $G$.
Examples of constraints on the distribution of $G$ are a fixed mean $\langle G \rangle$, an upper bound on the eigenvalues, or a fixed mean absolute value $\langle |G| \rangle$.

\begin{proof}
Consider any state $\rho_0$ that satisfies the constraint ${\cal C}$ on the distribution of $G$.
Given an arbitrary measurement described by a  POVM  $\{M_{\hat\phi}\}$, we obtain an average phase distribution $\bar p(\theta)$.
Using Lemma~\ref{lem1}, we find that there exists a covariant  POVM  $\{\overline{M}_{\hat\phi}\}$ such that the same probability distribution is obtained with the same state, $\rho_0$.
Next, using Lemma~\ref{lemz}, there exists a state without degeneracy, $\rho_0^{(s)}$, such that the nondegenerate canonical measurement on $\rho_0^{(s)}$ produces the same phase distribution, and the distribution of $G^{(s)}$ is the same as the distribution of $G$ for $\rho_0$.

Therefore, the distribution of $G^{(s)}$ for $\rho_0^{(s)}$ still satisfies the same constraint ${\cal C}$.
Furthermore, because the probability distribution for the canonical measurement $p_{C^{(s)}}(\theta|\rho_0^{(s)})$ is equal to the average phase distribution $\bar p_M(\theta|\rho_0)$, any measure of the concentration of the probability distribution is unchanged.
Because any value of the concentration that can be obtained for the arbitrary measurement under constraint ${\cal C}$ can also be obtained for the concentration of the canonical phase distribution under the same constraint, the arbitrary measurement must satisfy the same bound as the canonical measurement.
\end{proof}

\subsection{Analytic bounds via an entropic uncertainty relation}
\label{sec:anlyt}
It is possible to obtain a number of bounds by using entropic uncertainty relations.
The entropic uncertainty relation for canonical phase measurements and a nondegenerate shift generator $G$ is given by \cite{jmodopt,bbm}
\begin{equation}
\label{eq:entin}
H(\Theta)+H(G) \ge \ln 2\pi .
\end{equation}
This can then be used to obtain bounds on the {\ramse} \cite{rapid}.
 In particular, combining Eqs.~(\ref{entlength}), (\ref{entvar}) and (\ref{eq:entin}) yields
\begin{equation} \label{entbound}
\delta\hat\Phi \geq (2\pi e)^{-1/2} e^{H(\Theta)} \geq (2\pi/ e)^{1/2} e^{-H(G)}.
\end{equation}

We first specialise to the case where the eigenvalue spectrum $S$ includes all nonnegative integers, so we denote the generator by $N$.
The entropy for fixed mean number is maximised for the thermal (negative exponential) distribution.
By a straightforward calculation, one can show that this results in the inequality
\begin{equation}
H(N) \le \ln \langle N+1 \rangle + \langle N \rangle \ln (1+1/\langle N \rangle).
\end{equation}
Because $x \ln (1+1/x)<1$, this yields (for finite expectation values)
\begin{equation}
H(N) < \ln \langle N+1 \rangle +1.
\end{equation}

 Substitution into Eq.~(\ref{entbound}) then gives
\begin{equation}
\label{eq:heis}
\delta\hat\Phi > \frac{k_A}{\langle N+1 \rangle},
\end{equation}
where $k_A=\sqrt{2\pi/e^3}\approx 0.5593$ (defined in the Introduction).
Using Theorem~\ref{mainthm}, this result holds for the  {\ramse} for all possible phase measurements, and for any shift generator with nonnegative integer eigenvalues. 
 Recall that, because the {\mse} is a linear measure, the {\rmse} of the average distribution is equivalent to the {\ramse} [see Eqs.~\eqref{eq:aveq} and \eqref{eq:aveq2}]. 

We can also use this result to infer the result in the more general case where there is some lower bound $g$ on the eigenvalues of $G$.
Then we can take $G=g\openone+N$, so $\langle N\rangle = \langle G-g \rangle$.
Then one obtains
\begin{equation}
\delta\hat\Phi > \frac{k_A}{\langle G-g+1 \rangle}.
\end{equation}
Note that, for this result, it is not necessary for the spectrum $S$ to include all integers above $g$.
This is because, in minimising $\delta\hat\Phi$ for given $\langle G-g \rangle$, removing some integers restricts the possible states, and therefore can only increase the {\amse}.

An alternative restriction that one may wish to consider is, instead of a fixed mean, a fixed  mean of the absolute value, $\langle | G | \rangle$.
This is of particular interest in the case of angular momentum, where $G=J$.
Then fixed $\langle | J | \rangle$ corresponds to a mean absolute value of the angular momentum.
The maximum entropy for fixed $\langle | G-g | \rangle$,  where $g$ is any real number,  can be obtained by finding 
a critical point of the variational quantity
\begin{equation}\label{Lambda}
\Lambda = -\sum_{n\in S} p_n \ln p_n -\alpha\sum_{n\in S} p_n - \beta \sum_{n\in S} |n-g|p_n ,
\end{equation}
where $\alpha$ and $\beta$ are variational parameters. 
 As shown in Appendix \ref{detail}, this yields
\begin{equation} \label{eq:nonint}
H(G) < \ln (2\langle|G-g| \rangle+1) +1.
\end{equation}
Substitution in Eq.~(\ref{entbound}) then gives 
\begin{equation}
\label{eq:abheis}
\delta\hat\Phi > \frac{k_A}{\langle 2|G-g|+1 \rangle}.
\end{equation}
Once again we note that this result holds both when $S$ includes all integers, so $G=J$, and when $S$ does not include all integers.
In the latter cases, the maximum entropy distribution can not be obtained exactly, but it still provides a bound.

For a given state, one can adjust the value of $g$ in order to maximise this lower bound.
The optimal value is the median; that is, the value such that there is equal probability for eigenvalues above and below $g$.

Another restriction on the distribution that can be considered is a finite range of eigenvalues.
For example, with number we have a minimum eigenvalue of $0$, and can place an upper bound of $n_{\rm max}$ on the eigenvalues.
Then the entropy is bound as
\begin{equation}
H(G) \le \ln(n_{\rm max} +1),
\end{equation}
because the maximum entropy is for the flat distribution.
 Then, combining with \eqref{entbound} gives 
\begin{equation}
\delta\hat\Phi > \frac{\sqrt{2\pi/e}}{n_{\rm max}+1}.
\end{equation}
In the specific case of the Holevo variance, there is a well-known result for canonical measurements \cite{luis,WisKil97},
\begin{equation}
\delta_H \hat\Phi  \geq  \tan \left(\frac{\pi}{n_{\rm max}+2}\right).
\end{equation}
This result is achievable for arbitrary $n_{\rm max}$.
Using our Theorem, this result also holds for the average distribution for arbitrary measurements.
Furthermore, because the Holevo variance is a convex functional of the probability distribution,
this bound holds for the root-mean value of the Holevo variance (averaging over phase shifts).
 
\section{Optimal bounds via numerical calculations}
\label{sec:numer}
The bound in Eq.~\eqref{eq:heis} has a scaling constant of $k_A=\sqrt{2\pi/e^3}\approx 0.5593$.
In contrast, based on the asymptotic result for Holevo variance \cite{berrythesis,bandilla}, we expect $\delta\hat{\Phi}\gtrsim k_C /\langle N\rangle$ with $k_C =2(-z_A/3)^{3/2}\approx 1.3761$ for large $\langle N\rangle$, where $z_A$ is defined in Sec.~\ref{sec:hol}.
This indicates that the scaling constant of the bound in Eq.~\eqref{eq:heis} is not optimal, and suggests the conjecture \cite{rapid}
\begin{equation}
\label{eq:better}
\delta\hat\Phi > \frac{k_C}{\langle N+1 \rangle}.
\end{equation} 
In order to test  this conjecture, we solved  the variational problem to find the minimum
value of the {\ramse} or Holevo variance as a function of $\langle N \rangle$.
 The results supporting this conjecture are given in this section. 
In Sec.~\ref{sec:asymp} the conjecture is proved analytically for the special case of the asymptotic limit $\langle N\rangle\rightarrow\infty$.

\subsection{Holevo variance} 
The case of the Holevo variance is simplest, because the problem is to maximise $|\langle e^{i\Theta} \rangle|$.
Given a state
\begin{equation}
\ket{\psi} = \sum_{n=0}^{\infty} \psi_n \ket{n},
\end{equation}
we have
\begin{equation}
\label{eq:simhol}
|\langle e^{i\Theta} \rangle| = \left| \sum_{n=0}^{\infty} \psi_{n+1}\psi^*_{n} \right| .
\end{equation}
Note that we can upper bound this expression via
\begin{equation}
 \left| \sum_{n=0}^{\infty} \psi_{n+1}\psi^*_{n} \right|  \le  \sum_{n=0}^{\infty} |\psi_{n+1}\psi_{n}| .
\end{equation}
The normalisation and $\langle N \rangle$ are unaffected by replacing the coefficients $\psi_{n}$ with their absolute values.
Therefore, for maximisation of $|\langle e^{i\Theta} \rangle|$, we can always take $\psi_{n}$ to be real and nonnegative.

From the above, the variational problem is thus to find a critical point of
\begin{equation}
\Lambda = \sum_{n=0}^{\infty} \left( \psi_n\psi_{n+1} - \alpha \psi_n^2 -\beta n \psi_n^2\right) ,
\end{equation}
where $\alpha$ and $\beta$ correspond to normalisation and mean photon number constraints. 
The variational condition $\partial \Lambda /\partial \psi_n = 0$ leads directly to the eigenvalue equation
 \begin{equation} \label{rr}
 \psi_{n-1}+\psi_{n+1} = 2(\alpha +\beta n)\psi_n
 \end{equation} 
for $n\ge 1$, and $\psi_{1} = 2(\alpha +\beta n)\psi_0$.
To avoid the need to specify a different equation for $n=0$, we can simply define $\psi_{-1}:=0$.

\subsection{Root-mean-square error} 

The problem for minimising the {\ramse} is somewhat more difficult, because we do not have a simple expression like Eq.~\eqref{eq:simhol}.
However, any well-behaved function (i.e., satisfying the Dirichlet conditions) can be expanded in a Fourier series on the interval $[-\pi,\pi]$ as
\begin{equation}
f(\theta) = \sum_{m=-\infty}^{\infty} z_m e^{im\theta}.
\end{equation}
For $m\ge 0$, the expectation values of the exponentials are given by
\begin{equation}
\langle e^{im\Theta} \rangle = \sum_{n=0}^{\infty} \psi_{n+m}\psi^*_{n} .
\end{equation}
For $m<0$, the expectation values are just the complex conjugate of those for positive $m$.

Unlike the case of the Holevo variance, it is not obvious at first sight that we can take the state coefficients to be real.  However, if $f(\theta)$ is real, and symmetric about $\theta=0$, then $z_{-m}=z_{m}^*=z_m$.
Therefore the expectation value of $f(\Theta)$ is given by
\begin{equation}
\langle f(\Theta) \rangle =  \sum_{m,n=0}^{\infty} \psi_{m}^* Z_{mn}\psi_{n} 
\end{equation}
where $Z$ is the real symmetric matrix with coefficients $Z_{mn}:=z_{|m-n|}$.

The variational problem is then to find a critical point of
\begin{align}
&\Lambda = \langle f(\Theta) \rangle - \alpha -\beta \langle N \rangle \nn
&= \sum_{m,n=0}^\infty \psi_m^* \left[ Z_{mn} - (\alpha + \beta n)\delta_{mn}\right]\psi_n .
\end{align}
The variational condition leads to
\begin{equation}
\sum_{n=0}^{\infty} Z_{mn}   \psi_{n} = (\alpha+\beta m)\psi_{m}.
\end{equation}
This equation is solved as an eigenvalue equation with $\alpha$ as the eigenvalue.
Because the corresponding matrix $Z-\beta N$ is real and symmetric in the number state basis, the eigenvectors are real in this basis (up to a global phase factor).
This means that the state coefficients can indeed be taken to be real.

In the specific case of $f(\theta)=\theta^2$, the Fourier series is
\begin{equation}
\label{eq:four}
\theta^2 = \frac{\pi^2}{3}+ 4\sum_{m=1}^{\infty} \frac{(-1)^m}{m^2} \cos(m\theta).
\end{equation}
We then obtain the eigenvalue equation
\begin{equation}
\left(\frac{\pi^2}{3}-\beta m\right) \psi_m+ \mathop{\sum_{n=-m}^{\infty}}_{n\ne 0} \frac{2 (-1)^n}{n^2}  \psi_{n} = \alpha\, \psi_{m}.
\end{equation} 
Numerical solution of this eigenvalue equation is difficult, because there are an infinite number of Fourier coefficients.
The problem can be truncated at some maximum number, but solution still requires finding the eigenvalues of a full matrix.
In contrast, the problem for the Holevo variance is sparse, and can therefore be solved much more efficiently.

\subsection{Bounding the {\amse}} 
As we are interested in testing a lower bound on the {\amse}, we can alternatively use an expression with a finite number of Fourier coefficients, but that forms a lower bound on $\theta^2$. 
One alternative is to use $f_1(\theta):=2(1-\cos\theta)$, which is the same optimisation problem as for the Holevo variance.
To show $f_1(\theta)\le \theta^2$ on $[-\pi,\pi]$, we can use a Taylor expansion to third order with the Lagrange form of the remainder
\begin{equation}
f_1(\theta) = \theta^2 + \frac{f_1^{(3)}(\xi)}{3!}\theta^3 = \theta^2 - \frac 13 \theta^3\sin\xi ,
\end{equation}
where  $\xi\in[0,\theta]$. 
Because $\sin\xi$ has the same sign as $\theta^3$, the remainder term is negative, and $f_1(\theta)\le \theta^2$.

The drawback to this alternative is that it yields results that do not satisfy the conjectured lower bound.
We therefore use a higher-order approximation given by
\begin{equation}
f_2(\theta) := \frac 52 - \frac 83 \cos\theta + \frac 16 \cos2\theta.
\end{equation}
Again expanding in a Taylor series,
\begin{align}
f_2(\theta) &= \theta^2+\frac{f_2^{(3)}(\xi)}{3!}\theta^3 = \theta^2 - \frac 83 \theta^3 (1-\cos\xi)\sin\xi \nn
&\le \theta^2.
\end{align}
We can also obtain an upper bound using (see Appendix \ref{sec:tighter})
\begin{align}
f_3(\theta)&:= (\pi^2/4-1)[ 2(1-\cos\theta) - (1-\cos2\theta)/2 ]
\nn &\quad +2(1-\cos\theta).
\end{align}
In the following we will use $(\delta_m \hat{\Phi})^2 := \langle f_m(\Theta) \rangle$ for $m\in\{1,2,3\}$.

\subsection{Numerical results} 

The minimal Holevo variance, as well as the minimal values of $\langle \Theta^2 \rangle$ and $\delta_2\hat{\Phi}$, have been determined by numerically solving the eigenvalue equations.
In each case, a number cutoff was used that was about $10$ times the value of $\langle N \rangle$, or 100 for small $\langle N\rangle$.
At this point the magnitude of the state coefficients had fallen to less than $1/10^6$ of  the  maximum value, and increasing the cutoff beyond this did not alter the results by more than $1$ part in $10^6$.
For the results for $\langle \Theta^2 \rangle$, the maximum $\langle N \rangle$ was about $5000$, due to the difficulty in finding eigenvalues of a full matrix.
In contrast, for the Holevo variance and for $\delta_2\hat{\Phi}$, the maximum $\langle N \rangle$ was over $10^6$.

\begin{figure}[!t]
\centering
\includegraphics[width=0.47\textwidth]{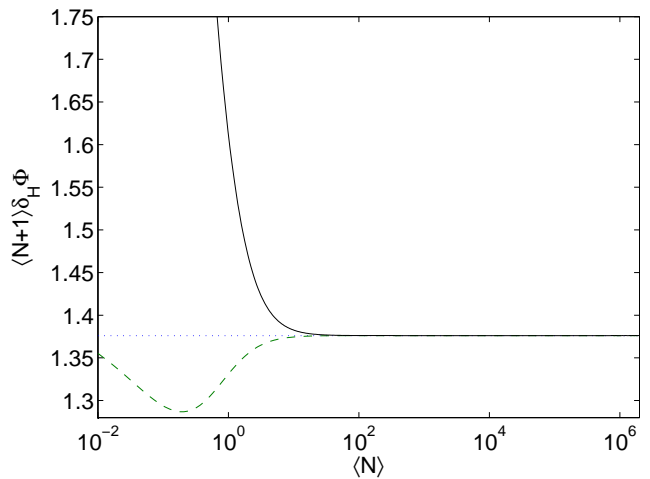}
\caption{Minimum possible value of $\langle N+1\rangle \delta_H \hat{\Phi}$, plotted as a function of $\langle N\rangle$ (solid curve).
The case where $\delta_1\Phi$ is used instead of $\delta_H \hat{\Phi}$ is shown as the dashed curve (green).
The asymptotic value of $ k_C \approx 1.3761$ is shown as the horizontal dotted line (blue).}
\label{fig:numerical}
\end{figure}

\begin{figure}[!t]
\centering
\includegraphics[width=0.47\textwidth]{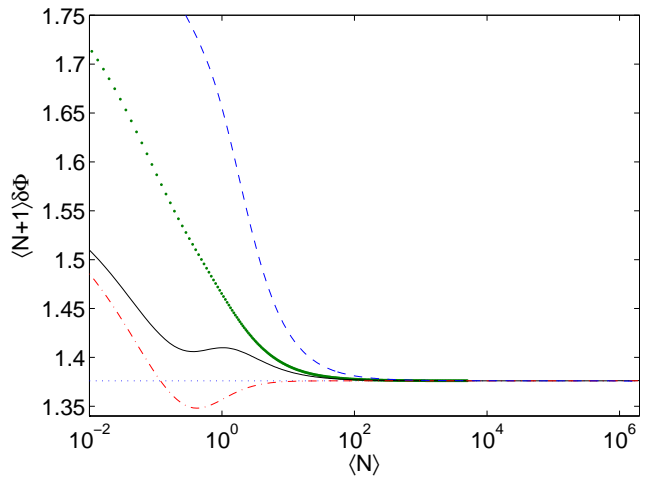}
\caption{Minimum possible value of $\langle N+1\rangle \delta \hat{\Phi}$, plotted as a function of $\langle N\rangle$.
The dotted curve (green) shows the values obtained for $\delta \hat{\Phi}$ (i.e., the {\ramse}).
The solid curve (black) uses $\delta_2\hat{\Phi}$ instead of $\delta \hat{\Phi}$.
The dash-dotted curve (red) is the lower bound using $\arccos[1-(\delta_1\hat{\Phi})^2/2]$.
The dashed curve (blue) is the upper bound using $\delta_3\hat{\Phi}$ calculated for the state that minimises $\delta_1\hat{\Phi}$.
The asymptotic value of $ k_C \approx 1.3761$ is again shown as the horizontal dotted line (blue).}
\label{fig:numerical2}
\end{figure}

The results for the Holevo variance are given in Fig.~\ref{fig:numerical}.
In this figure the square root of the Holevo variance is plotted multiplied by $\langle N+1 \rangle$.
Therefore, if $k_C/\langle N+1 \rangle$ provides a lower bound to the {\ramse}, the curve should 
be above $k_C$ (also shown in the figures).
It is clear from the figure that the numerical results indicate that $k_C/\langle N+1 \rangle$ provides a strict lower bound to $\delta_H \hat{\Phi}$.
In this figure $\delta_1\hat{\Phi}$ is also shown, and  $\delta_1 \hat{\Phi}< k_C/\langle N+1 \rangle$ in the range shown. 

The results calculated for $\delta \hat{\Phi}$ are shown in Fig.~\ref{fig:numerical2}.
It can be seen that these results are also above the line for $k_C$, indicating that $\delta \hat{\Phi}> k_C/\langle N+1 \rangle$.
One would like to provide more easily calculated lower bounds on $\delta \hat{\Phi}$ to test this inequality more thoroughly.
It is clear that $\delta_1\hat{\Phi}$ is not useful for this  purpose,  because the curve in Fig.~\ref{fig:numerical} is below $k_C$.
It is also possible to obtain a tighter lower bound on $\delta\hat{\Phi}$ using $\delta_1\hat{\Phi}$ (see Appendix \ref{sec:tighter}), but this curve is still not above $k_C$ for all $\langle N\rangle$.

A better lower bound to $\delta \hat{\Phi}$ is $\delta_2\hat{\Phi}$, which is also shown in Fig.~\ref{fig:numerical2}, and is above the $k_C$ line.
This quantity can be calculated more rapidly and reliably than $\delta \hat{\Phi}$, and results are given  up  to $\langle N \rangle \approx 2\times 10^6$.
This provides further numerical evidence that $\delta \hat{\Phi}> k_C/\langle N+1 \rangle$.
Results were also calculated for $\langle N \rangle$ down to about $10^{-6}$.
These are not shown in the figures, but the curves that are above $k_C$ do not cross below $k_C$.

\subsection{Angular momentum calculations} 
We have also calculated the corresponding results with a fixed value of $\langle |J| \rangle$.
The variational problem is exactly the same as before, except now we sum over positive and negative values of $j$ (as opposed to $n$), and replace $n$ with $|j|$.
That is, the variational problem is to find a critical point of
\begin{equation}
\Lambda = \langle f(\Theta) \rangle - \alpha -\beta \langle |J| \rangle.
\end{equation}
As before, for a real function $f$ symmetric about zero we can assume that the state coefficients are real, so the variational condition yields
\begin{equation}
\sum_{m=-\infty}^{\infty} a_m  \psi_{j+m} = (\alpha+\beta j)\psi_{j}.
\end{equation}
In the case of $f(\theta)=\theta^2$, we obtain the eigenvalue equation
\begin{equation}
\frac{\pi^2}{3}\psi_j+ \mathop{\sum_{m=-\infty}^{\infty}}_{m\ne 0} \frac{2 (-1)^m}{m^2} \psi_{j+m} = (\alpha+\beta |j|)\psi_{j}.
\end{equation}
The eigenvalue equation for the case of $f_1(\theta)$ is
\begin{equation} \label{rr2}
 \psi_{j-1}+\psi_{j+1} = 2(\alpha +\beta |j|)\psi_j .
\end{equation}
We will not consider $f_2(\theta)$ for this problem.

The results for $\delta\Phi$, $\delta_H\Phi$, and $\delta_1\Phi$ were all determined numerically, and the results are shown in Fig.~\ref{fig:figlbi}.
It will be shown in the next section that the asymptotic optimal value for $\delta_1\hat{\Phi}$ is
\begin{equation}
\delta_1\hat{\Phi} \sim \frac{k'_C}{\langle 2|J|+1 \rangle},
\end{equation}
with $k'_C = 4(-z'_A/3)^{3/2}\approx 0.7916$, where $z'_A$ is the first zero of the derivative of the Airy function.
We have therefore plotted the results for $\delta \hat{\Phi}$ multiplied by $\langle 2|J|+1\rangle$ in Fig.~\ref{fig:figlbi}.
It can be seen in this figure that all the results are above $k'_C$, supporting the conjecture that there is strict inequality with the scaling constant $k'_C$.

\begin{figure}[!t]
\centering
\includegraphics[width=0.47\textwidth]{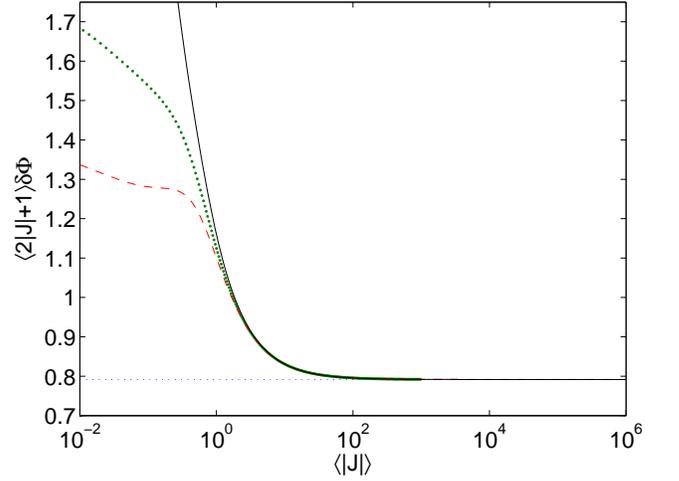}
\caption{Minimum possible value of $\langle 2|J|+1\rangle \delta \hat{\Phi}$ is plotted as a function of $\langle N\rangle$ as the dotted curve (green).
The minimum value of $\langle 2|J|+1\rangle \delta_H \hat{\Phi}$ is shown as the solid curve (black), and $\langle 2|J|+1\rangle \delta_1 \hat{\Phi}$ is 
shown as the dashed curve (red).
The asymptotic value of $ k'_C \approx 0.7916$ is shown as the horizontal dotted line (blue).}
\label{fig:figlbi}
\end{figure}

\section{Asymptotic expansions}
\label{sec:asymp}
\subsection{Holevo variance}  \label{sec:asymp-HV}
In the specific case of the Holevo variance, it is possible to obtain analytic results in terms of Bessel functions to provide further support to the conjecture that $\delta_H \hat{\Phi}> k_C/\langle N+1 \rangle$.
The recurrence relation \eqref{rr} has a known solution in terms of Bessel functions \cite{bandilla}.
Bessel functions of the first kind satisfy the recurrence relation $J_{k-1}(z) +J_{k+1}(z)=(2k/z)J_k(z)$.
Therefore the solution is of the form
\begin{equation}
\psi_n(x,z) = A J_{x+n+1} (z),
\end{equation}
with $x:=\alpha/\beta-1$, $z:=1/\beta$.
Bessel functions of the second kind can be ignored, because they diverge for large values of the order.
The condition that $\psi_{-1}=0$ implies the restriction
\begin{equation} \label{zero}
J_x(z) = 0
\end{equation}
on the parameter $z$, thus confining its allowed values to the (countable) set of zeroes of $J_x$. 

To obtain the smallest Holevo variance for a given mean photon number, we wish to take the solution for the \emph{largest} value of $\alpha$.  This corresponds to the largest solution of Eq.\ \eqref{zero} in terms of $x$ for given $z$.  Conversely, for given $x$ we want the first positive zero of $J_x$.
The normalisation constraint yields
\begin{equation}
 A^{-2} = \sum_{n=0}^\infty [J_{x+n+1}(z)]^2 = \sum_{k=1}^\infty [J_{x+k}(z)]^2 ,
\end{equation}
and hence one has
\begin{align}
\langle N\rangle &=A^2 \sum_{n=0}^\infty n\,[J_{x+n+1}(z)]^2 = A^2\sum_{k=1}^\infty (k-1)[J_{x+k}(z)]^2\nn
\label{nbar}
 &=\frac{\sum_{k=1}^\infty k\,[J_{x+k}(z)]^2 }{ \sum_{k=1}^\infty [J_{x+k}(z)]^2} -1 , \\
 \label{sharp}
\langle e^{i\Theta}\rangle &= A^2 \sum_{n=0}^\infty J_{x+n+1}(z)\,J_{x+n+2}(z) \nn
&= \frac{\sum_{k=1}^\infty J_{x+k}(z)\,J_{x+k+1}(z) }{\sum_{k=1}^\infty [J_{x+k}(z)]^2}.
\end{align}

Using Eq.\ \eqref{rr}, we have
\begin{equation}
\label{eigsol}
\langle e^{i\Theta}\rangle = (\alpha+\beta\langle N \rangle) = (x+\langle N \rangle + 1)/z.
\end{equation}
Up until this point, these results for the Bessel functions are the same as those of Ref.~\cite{bandilla}.
Reference \cite{bandilla} then uses an approximation in terms of Airy functions.
 We have determined more accurate results using formulae for sums of Bessel functions (see Appendix \ref{asymp}).
We find that
\begin{equation}
\label{eq:series}
| \langle e^{i\Theta}\rangle|^{-2}-1 = \sum_{k=1}^5 \frac{b_{2k}}{\langle N +1\rangle ^{2k}} + O\left( \frac 1{\langle N +1\rangle ^{12}} \right),
\end{equation}
where $b_2 = k_C^2$, and $b_2$ to $b_{10}$ are all positive and close to 2.
The fact that each $b_j$ that has been calculated is positive strongly supports  the conjecture that the Holevo variance is strictly lower bounded by the first term.

\subsection{Upper bounding the optimal mean-square error}  \label{sec:asymp-{\mse}}
It would be desirable to obtain a similar approximation for the exact {\ramse} $\delta\Phi$.
However, the eigenvalue equation does not have any solution in terms of elementary functions that we have been able to find.
Even the lower bounding quantity $\delta_2\Phi$ yields an eigenvalue equation that does not appear to have an analytic solution.
However, we can place an upper bound on the optimal value of $\delta\Phi$, using
\begin{align}
\label{eq:inva}
\theta^2 \le f_3(\theta),
\end{align}
for $\theta\in [-\pi,\pi]$.
We can calculate $\delta_3\hat{\Phi}$, except for the state that minimises $\delta_1\hat{\Phi}$.
This value is shown in Fig.~\ref{fig:numerical2} for comparison with $\delta\hat{\Phi}$.

Using the properties of Bessel functions, this leads to the result that the optimal value of $\delta\hat{\Phi}$ satisfies (see Appendix \ref{asymp})
\begin{equation}
\label{eq:upbnd}
(\delta\hat{\Phi})^2 \le \frac{k_C^2}{\langle N+1 \rangle^2} + O\left(\frac 1{\langle N+1 \rangle^3}\right).
\end{equation}
This means that, asymptotically, the optimal value of $\delta\hat{\Phi}$ cannot be larger than $k_C/\langle N+1 \rangle$.
Because $\delta\hat{\Phi}$ cannot be smaller than $\delta_H\hat{\Phi}$ except for higher-order terms [see Eq.~\eqref{eq:vaholin}], 
this means the optimal
$\delta\hat{\Phi}$ must be asymptotically equal to $k_C/\langle N+1 \rangle$ [i.e., 
$k_C$ is the largest value of $k$ for which Eq.~\eqref{eq:heis0} can be true].

\subsection{Angular momentum calculations}  \label{sec:asymp-AM}
Next we consider the problem with fixed $\langle |J|\rangle$.
Recall that the variational problem yields an eigenvalue problem given in Eq.~\eqref{rr2}.
This is solved by taking \cite{dariano}
\begin{equation}
\psi_j(x,z) = A_1 \,J_{x+j}(z) ,
\end{equation}
for $j\ge 0$, and
\begin{equation}
\psi_j(x,z) = A_2 \,J_{x-j}(z) ,
\end{equation}
for $j\le 0$.
In this case we take $x:=\alpha/\beta$, $z:=1/\beta$.
We again may ignore Bessel functions of the second kind, because they diverge.
The restriction that the solutions coincide for $n=0$ means that $A_1=A_2=A$, and
\begin{equation}
\psi_j(x,z) = A \,J_{x+|j|}(z),
\end{equation}
for all $j$.
The condition that the recurrence relation holds for $j=0$ means that
\begin{equation}
J_{x+1}(z) = \frac xz J_{x}(z) = \frac 12 [J_{x-1}(z)+J_{x+1}(z)].
\end{equation}
This implies
\begin{equation}
J_{x-1}(z) - J_{x+1}(z)=0.
\end{equation}
Then, using $[J_{x-1}(z) - J_{x+1}(z)]/2=J'_{x}(z)$, this means we must have $J'_{x}(z)=0$.

Performing series expansions similar to that for the first case, gives (see Appendix \ref{asymp2})
\begin{equation}
\label{eq:spinser}
2(1-| \langle e^{i\Theta}\rangle|) = \sum_{k=2}^{9} \frac{d_{k}}{\langle 2|J|+1 \rangle ^{k}}  + O\left( \frac 1{\langle 2|J|+1 \rangle^{10}} \right),
\end{equation}
where $d_2 = {k'_C}^2$, and  coefficients  up to $d_5$ are positive, but $d_6$ is negative.
This  strongly supports the numerical results that the strict inequality
\begin{equation}
\delta_1\hat{\Phi} \ge \frac{k'_C}{\langle 2|J|+1 \rangle},
\end{equation}
holds.
In turn, because $\delta\hat{\Phi} \ge \delta_1\hat{\Phi}$, this also supports the conjecture that the inequality holds for $\delta\hat{\Phi}$.
In addition, $\delta_H\hat{\Phi} \ge \delta_1\hat{\Phi}$, so this supports the conjecture that the inequality holds for $\delta_H\hat{\Phi}$.

Similarly to the case for fixed $\langle N\rangle$, one can use Eq.~\eqref{eq:inva} to find a series expansion for an upper bound on 
the optimal value of $\delta\hat{\Phi}$, giving
\begin{equation}
(\delta\hat{\Phi})^2 \le \frac {{k'_C}^2}{\langle 2|J|+1 \rangle^2}+ O\left(\frac 1{\langle 2|J|+1 \rangle^3}\right).
\end{equation}
This means that we have upper and lower bounds on the optimal $\delta\hat{\Phi}$, showing that it is asymptotically equal to $k'_C/\langle 2|J|+1 \rangle$.

\section{Scaling with number of probe states}
\label{sec:probe}
Another question is what the scaling of the lower bound is if there are $m$ identical probe states.
Normally it is expected that the {\mse} will scale like $1/\sqrt{m}$ if there are $m$ copies of the state.
This is because, for estimates formed by the average of the individual estimates, the standard error scales as $1/\sqrt{m}$.
Similarly, the Cram\'er-Rao bound for $m$ identical probe states yields the following bound for  estimates  that are unbiased  
(in the standard statistical sense, not what we have called {\unb} in Sec.~\ref{sec:hol}) \cite{Helstrom,holcov}:
\begin{equation}
\label{eq:cram}
\Delta_\phi \hat{\Phi} \ge \frac{1}{2\sqrt{m} \Delta N}.
\end{equation}
We will call this the Helstrom-Holevo bound.

Because of these results one might expect that one could derive a lower bound 
to the uncertainty in terms of $\langle N\rangle$ of the form $k/(\sqrt{m}\langle N +1\rangle)$.  
On the other hand, directly using the above methods yields a lower bound of
\begin{equation}
\label{mNbound}
\delta \hat{\Phi} \ge \frac{k}{\langle mN +1\rangle},
\end{equation}
because the overall average number is $m\langle N \rangle$.
Recall that we have proven this inequality for $k=k_A$, and have extremely strong 
numerical evidence for the inequality for $k=k_C$.

We can prove that there is no lower bound scaling as $1/(\sqrt{m}\langle N +1\rangle)$ in the following way.
Let $m\in\mathbb{N}$, $\langle N \rangle$, and $\delta>0$ be given.
We use $\mu$ for the required value of $\langle N \rangle$, to avoid confusion with intermediate states we use in this discussion with different values of $\langle N \rangle$.

Let $\ket{\chi_{n-1}}$ be the state with the minimum phase uncertainty for mean number $\langle N \rangle=n-1$,
and let $\ket{\chi'_{n-1}}$ be the corresponding state with the same amplitudes, but shifted up by one.
This means that there is no vacuum component, the phase uncertainty is unchanged, and $\langle N \rangle=n$.
We are considering small $\delta$ and large $m$, so we expect that $\mu^\delta \le m$.
In that case, we take $n=(m\mu)^{1/(1+\delta)}$, and consider $m$ copies of
\begin{equation}
\ket{\psi} = \sqrt{1-\mu/n}\ket{0}+ \sqrt{\mu/n}\ket{\chi'_{n-1}}.
\end{equation}
For this state, $\langle N \rangle = \mu$.
Now consider a phase measurement that first distinguishes between $\ket 0$ and $\ket{\chi'_{n-1}}$ on all copies of the state.
If the $\ket{\chi'_{n-1}}$ result is found, then a canonical phase measurement is performed.

The probability of getting the $\ket{\chi'_{n-1}}$ result is $\mu/n$.
For $m$ repetitions, the probability of projecting every single copy onto the state $\ket{0}$ is $(1-\mu/n)^m\le \exp(-m\mu/n) = \exp(-(m\mu)^{\delta/(1+\delta)})$.
This probability scales exponentially in $m\mu$ and may be ignored for asymptotically large $m\mu$. 
The phase uncertainty is therefore (up to an exponentially small correction) no more than that for $\ket{\chi'_{n-1}}$, which is
\begin{align}
\delta\hat{\Phi} &= k_C/n + O(1/n^2) \nn
&= k_C/(m\mu)^{1/(1+\delta)} + O(1/(m\mu)^{2/(1+\delta)}).
\end{align}

For $\mu^\delta > m$, we can just take $n=\mu$, and $\ket{\psi} = \ket{\chi'_{n-1}}$.
In this case we have $1/(m\mu)^{1/(1+\delta)}\ge 1/\mu$.
Therefore, considering just the uncertainty for a single copy of the state gives
\begin{align}
\delta\hat{\Phi} &= k_C/\mu + O(1/\mu^2) \nn
&< k_C/(m\mu)^{1/(1+\delta)} + O(1/(m\mu)^{2/(1+\delta)}).
\end{align}
This provides an upper bound to the uncertainty for $m$ copies of the state.

Therefore, we find that, for any $\delta>0$, $m\in\mathbb{N}$ and $\mu=\langle N \rangle$, we can find a state such that
the uncertainty is no greater than $k_C/(m\mu)^{1/(1+\delta)}$ to leading order.
Because we can choose any $\delta>0$, this means that, for fixed $\langle N \rangle$, the lower bound to the scaling must be arbitrarily close to $1/m$.

This result is counterintuitive, because for a state that does not depend on $m$, 
the uncertainty can be expected to scale as $1/\sqrt{m}$, similarly to the Helstrom-Holevo bound \eqref{eq:cram}. 
However, the Helstom-Holevo bound, in terms of $m$ and  $\Delta N$, 
holds even for states that depend on $m$.
Similarly, a bound in terms of $m$ and $\langle N \rangle$ must hold for states that are chosen based on $m$.
We have shown that the potential dependence of states upon $m$ means that it is not possible to obtain 
a universal bound that scales as $1/\sqrt{m}$ for given $\langle N \rangle$. 

\section{Papers claiming violation of Heisenberg limit}
\label{sec:viol}
In the following, we present some recent measurement schemes claiming violation of the Heisenberg limit.
We summarise the techniques used in these schemes, and explain why they appear to violate the Heisenberg limit.
We argue that the accuracy of these super-Heisenberg measurements should be considered illusory, primarily because they only work for a very restricted range of phase.

\subsection{Anisimov et al.}
Anisimov {\it et al.} \cite{ani} describe a noncovariant phase estimation method having a minimum {\rmse}
 \begin{equation} \label{ani}
\Delta^0_\phi \hat{\Phi} = \frac{1}{[\langle N\rangle(\langle N\rangle +2)]^{1/2}}  .
 \end{equation}
This quantity is for a particular phase shift, as opposed to the average over the phase shift, $\delta\hat{\Phi}$.
Also, the {\rmse} is here using a reference phase of $0$, rather than the reference phase of $\phi$ that we use (see Sec.~II~A). 
This result violates an alternative definition of the Heisenberg limit, given by Anisimov {\it et al.} as \cite{ani}
\begin{equation}
\label{eq:oldHeis}
\Delta^0_\phi \hat{\Phi} \geq 1/\langle N\rangle .
\end{equation}
First, it should be noted that this does not give a different power of $\langle N\rangle$, and does not change the scaling constant.
It only violates this form of the Heisenberg limit by an amount which is significant for small $\langle N\rangle$, and is of higher order for large $\langle N\rangle$.
In later work \cite{ani2}, they have modified their claim to that of achieving the Heisenberg limit.

In fact, it is easy to see that the above form (\ref{eq:oldHeis}) of the Heisenberg limit cannot be a strict limit for small $\langle N\rangle$.
For any $\langle N\rangle$ less than $1/\pi$ it must be violated, because the maximum {\rmse} possible is $\pi^2$.
For the same reason, any bound of the form $k/\langle N\rangle$ cannot hold for all $\langle N \rangle$.
It is for this reason that we have used $\langle N +1\rangle$ (or $\langle G+1 \rangle$ more generally).

Note from Eq.~(\ref{ani}) that the minimum {\rmse} satisfies
\begin{equation}
\Delta^0_\phi \hat{\Phi} >  \frac{1}{[\langle N\rangle(\langle N\rangle +2) +1]^{1/2}}  = \frac{1}{\langle N+1\rangle} .
\end{equation}
Hence, the {\ramse}, $\delta\hat{\Phi}$, trivially satisfies our analytical lower bound \eqref{eq:heis}.
However, the minimum value is below our conjectured best possible bound \eqref{eq:better}, by a factor of $k_C\approx 1.3761$ 
in the asymptotic limit. 
This does not contradict our conjectured bound, because the conjectured bound is for $\delta\hat{\Phi}$, whereas the above value is for a specific value of the phase shift.
This is an important aspect of our result.
It is possible to obtain smaller errors for specific values of the phase shift \cite{Hall12,tsang,gm,nair,hwprx}, but not when the average is taken over the phase.
That is, the noncovariance of the scheme in Ref.~\cite{ani} does allow beating our conjectured bound, in a small range of phase shifts about $\phi=0$, but only at the expense of worse phase resolution over the remainder of possible phase shift values.  

\subsection{Zhang et al.}

Zhang {\it et al.}~\cite{zhang} propose a superposition state with arbitrarily high phase sensitivity but finite average photon number.
They consider a two-mode Mach-Zehnder interferometer (MZI) system, with a probe state of the form
\begin{align} \label{zhang}
|\psi\rangle &:= \sum_{n\geq 1} c_n |\psi_n\rangle, \nn
|\psi_n\rangle &:= \frac{1}{\sqrt{2}}\left[ |n,0\rangle + |0,n\rangle\right] .
\end{align}
That is, $|\psi\rangle$ is a superposition of (mutually orthogonal) NOON states.

They use the quantum Cram\'{e}r-Rao bound (QCRB) to derive the ultimate limit to the uncertainty of phase measurement as
\begin{equation} \label{MSHL}
\Delta_\phi^0 \hat{\Phi} \ge \frac{1}{\sqrt{\langle N^2\rangle}} ,
\end{equation}
in contrast to Eq.~\eqref{eq:oldHeis}.
Also $N=N_a+N_b$ is the total photon number operator for the two modes $a$ and $b$,
rather than just the number operator  $N_a$  for the mode passing through the phase shift.  
They call  Eq.~(\ref{MSHL})  the ``proper'' Heisenberg limit, and Eq.~\eqref{eq:oldHeis} the ``generally accepted form'' of the Heisenberg limit.
This result is similar to the result given in a number of other works \cite{Hofmann,Hyllus}.
By choosing $c_n\propto n^{-3/2}$, they obtain $\langle N\rangle < \infty$ and $\langle N^2\rangle=\infty$, which gives $\Delta^0_\phi \hat{\Phi} \ge 0$.
They further claim in Sec.~V of Ref.~\cite{zhang} that this lower bound is achievable (i.e., that the uncertainty can be zero for finite $\langle N\rangle$).

An interesting feature of their result is that the Fisher information can be infinite for finite $\langle N\rangle$.
Therefore, it should not be expected that the QCRB can give a nontrivial lower bound on the uncertainty for fixed $\langle N\rangle$.
Furthermore, the Fisher information is infinite for all $\phi$.

However, there are some problems with the result presented.
First, they give no proof that the lower bound provided by the QCRB is achievable.
In many cases Fisher's theorem \cite{Fisher} allows the QCRB to be achieved asymptotically (i.e., with a scaling constant of $1/\sqrt{m}$ for $m$ probe states).
However, Fisher's theorem is not universally applicable, because it requires a unique maximally likely estimate \cite{durkin}.
In contrast, here the measurements will yield multiple maximally likely estimates.

Second, the form of the QCRB given is for unbiased measurements, but it is unclear how to perform an unbiased measurement here.
For biased measurements this lower bound does not hold.
In fact, the obvious measurement technique is biased, and will  only  yield zero error for $\phi=0$ and $\pi$, similar to the example in Sec.~\ref{sec:ex}.
This can be achieved with a very simple choice of state.

However, measurements that yield zero error  only  for isolated values of $\phi$ will not be useful.
 Further,  based on the results presented here,
the {\it average} performance of any two-mode MZI estimate must satisfy 
\begin{equation}
\delta \hat{\Phi} \geq \frac{k_A}{\langle N_a+1 \rangle},
\end{equation}
as a  consequence of Eq.~(\ref{eq:heis}).

\subsection{Rivas and Luis}
Rivas and Luis \cite{rivas} consider a linear phase estimation procedure that employs as the probe state the coherent superposition
\begin{equation}
|\psi\rangle = \mu |0\rangle + \nu |\xi\rangle
\end{equation}
of the vacuum $|0\rangle$ and a squeezed state $|\xi\rangle$. The authors consider the case with $\nu \ll 1$, $\mu \simeq 1$ and also assume that the phase shift is known to be small: $\phi \ll 1$. The fixed mean photon number of the probe state is then given by
\begin{equation}
\langle N \rangle = \nu^2 \bar{n}_{\xi}\, ,
\end{equation}
where $\bar{n}_{\xi}$ is the  (average)  number of photons in the squeezed state. 
Using conventional error propagation arguments, they find for this state 
\begin{equation}
(\Delta_\phi^0 \hat{\Phi})^2 \geq \frac{\nu^2}{4m\langle N \rangle^2},
\end{equation}
where $m$ is the number of repetitions of the measurement.
The lower bound here is arbitrarily 
below the usual Heisenberg limit by a factor $O(\nu^2)$. 

We note that similar results can be  obtained in a simplified scenario by employing the probe state
\begin{equation}
|\psi\rangle = \mu |0\rangle + \nu |\bar{n}_{\xi}\rangle\, ,
\end{equation}
where $|\bar{n}_{\xi}\rangle$ denotes a number state with $\bar{n}_{\xi}$ photons.
The interference fringes obtained from this state are high frequency, but low visibility.
A calculation using the error propagation formula based on the observable $X = |0\rangle\langle \bar{n}_{\xi}| + |\bar{n}_{\xi}\rangle\langle 0|$ yields
\begin{equation}
(\Delta_\phi^0 \hat{\Phi})^2 = \frac{(\delta X)^2}{|d\langle X \rangle/d \phi|^2} \approx \frac{\nu^2}{4 \langle N \rangle^2}\, .
\end{equation}
Taking $\nu \propto \langle N \rangle ^{1-p}$ gives
\begin{equation}
\Delta_\phi^0 \hat{\Phi} \propto \frac 1 { \langle N \rangle^{p}},
\end{equation}
which, in principle, gives an accuracy that scales arbitrarily well with $\langle N \rangle$ (for large $p$).

The problem with this  scheme  is that the high accuracy predicted by the error propagation formula is given by high frequency fringes with low visibility.
It would take a great deal of additional phase information to resolve the ambiguity in the fringes,
as well as many repetitions of the measurement to obtain a reasonable estimate of the observable $X$ so that the error propagation formula would become accurate.

The scheme presented in Ref.~\cite{rivas} is a little more complicated (including an analysis of the efficiency), but similar considerations apply.
A quadrature measurement is considered for a fixed phase, which means that the analysis essentially gives an estimate of the uncertainty for a given value of the phase.
As we have noted above, it is possible to obtain higher accuracy for a particular value of the phase shift.
For example, it is trivial to design a measurement that gives zero error for a single value of the phase shift.
The bound \eqref{eq:heis} must hold when averaging over the phase shift.

\subsection{Nonlinear interferometry}
\label{sec:nonlin}
A qualitatively different type of proposal for beating the Heisenberg limit is that based on nonlinear interferometry \cite{nonlinear,napolit}.
The basis of these proposals is that the generator of the phase shifts is nonlinear in the number operator.
For example, $G=N^q$ for some $q>1$.
It is then found that the phase uncertainty can scale as $1/\langle N \rangle^q$.
Subtleties involved in achieving such scalings are discussed in Ref.~\cite{hwprx}.

These proposals do not contradict the results presented here; they are just using the terminology differently \cite{zwierz}.
In Refs.~\cite{nonlinear,napolit}, the Heisenberg limit is given as $1/\langle N \rangle$, where $N$ is the number of 
particles. 
In contrast, here we give the Heisenberg limit in terms of the generator of the phase shifts.
That is, the bound is
\begin{align} \label{nonlinHL}
\delta \hat{\Phi} &\ge \frac {k}{\langle G+1\rangle} = \frac {k}{\langle N^q+1\rangle},
\end{align}
which  typically  scales as ${k}/{\langle N \rangle ^q}$.
Therefore the results do not violate the Heisenberg limit (\ref{nonlinHL}) given here.
In Refs.~\cite{nonlinear,napolit}, they call this limit the ``quantum limit'', rather than the Heisenberg limit.

\section{Limitations of the Cram\'er-Rao bound}
\label{sec:cr}
The Cram\'er-Rao bound for the {\rmse} $\Delta_\phi^0 \hat{\Phi}$ is often used as motivation for the Heisenberg limit, but it has limitations which mean that it does not provide a rigorous basis for the Heisenberg limit.
There are a number of different variations of the way the Cram\'er-Rao bound is used.
First, the classical Cram\'er-Rao bound (CRB), $1/\sqrt{mF_C(\phi)}$, is in terms of the classical Fisher information $F_C(\phi)$ of a specific probability distribution, so in quantum mechanics it is calculated for a given state and measurement.

Second, the \emph{quantum} Cram\'er-Rao bound (QCRB) replaces $F_C(\phi)$ by the quantum Fisher information, $F_Q(\phi)$ (corresponding to the classical Fisher information optimised over all quantum measurements), but  is still calculated for a given state \cite{qcrb}.
Third, the Helstrom-Holevo bound (HHB), as in Eq.~\eqref{eq:cram}, is optimised over both the quantum measurement and the quantum state, with the optimisation being for a given $\Delta N$.
Because these bounds use successively more optimisation, one has the ordering $CRB \geq QCRB \geq HHB$.
In particular, for any estimate that is {\it unbiased} for phase shift $\phi$, one has
\begin{equation} \label{chain}
\Delta_\phi^0 \hat{\Phi} \geq \frac{1}{\sqrt{mF_C(\phi)}} \geq \frac{1}{\sqrt{mF_Q(\phi)}}  \geq \frac{1}{2\sqrt{m}\Delta N} .
\end{equation}

%\subsection{Cram\'er-Rao does not give universal bounds for fixed $\langle N\rangle$}
The most obvious limitation in using the HHB is that it is a limit in terms of $\Delta N$, whereas the Heisenberg limit is in terms of $\langle N \rangle$.
This means that, for states with large uncertainty in $N$ as compared to the mean value, the HHB does not imply the Heisenberg limit.
This is taken advantage of in Refs.~\cite{zhang,rivas}.

A fixed value of $\Delta N$ is just a choice of constraint.
One could also consider optimisation for fixed $\langle N \rangle$, as a method to obtain the Heisenberg limit.
However, it is easily seen that there is no upper bound on the Fisher information for a given $\langle N \rangle$.
In the example of Zhang \textit{et al.}, they find a state with infinite Fisher information for finite $\langle N \rangle$ (see Sec.~VII~B).
 Note also that if there were such a bound, then the CRB would imply a $1/\sqrt{m}$ scaling for fixed $\langle N \rangle$, whereas we have found that such scaling is impossible (see Sec.~\ref{sec:probe}). 
The difficulty of using the CRB was also noted in Ref.~\cite{glm}.

\subsection{ Bias in phase estimation}

Another major factor that needs to be taken into account when considering the CRB and related bounds is that of bias. 
Note, for example, that the value of $\Delta N$ in Eq.~(\ref{chain}) can be arbitrarily small, whereas the {\rmse} cannot be larger than $\pi$.
It follows that any phase estimate must be biased for sufficiently small $\Delta N$.
In fact, one can show that covariant phase measurements cannot be unbiased for every phase shift value, in the sense needed for the QCRB and HHB.

 In particular, when considering the {\rmse} with reference phase $\phi_r$, $\Delta_\phi^{\phi_r} \hat{\Phi}$, one needs to define the bias function
\begin{equation} \label{bphi}
b_{\phi_r}(\phi):= \langle\hat\Phi\rangle_\phi^{\phi_r}-\phi,
\end{equation}
with
\begin{equation}
\langle\hat\Phi\rangle_\phi^{\phi_r}:= \int_{\phi_r-\pi}^{\phi_r+\pi}  d \hat{\phi} \, \hat{\phi} \, p( \hat{\phi} | \phi) .
\end{equation}
Then the CRB with bias is \cite{biasedcrb}
\begin{equation} \label{bias}
(\Delta_\phi^{\phi_r} \hat{\Phi})^2 \geq \frac{[1+b'_{\phi_r}(\phi)]^2}{mF_C(\phi) } +b_{\phi_r}(\phi)^2 .
\end{equation}
The QCRB and HHB in Eq.~(\ref{chain}) similarly generalise (see also \cite{Hel67}).

If one is to use the form of the CRB \emph{without} bias, then one needs $b_{\phi_r}(\phi)=0$ and $b'_{\phi_r}(\phi)=0$.
This is highly problematic if one is to consider the full range of values of $\phi$ with a fixed reference phase $\phi_r$.
This is because $\langle\hat\Phi\rangle_\phi^{\phi_r}$ would need to change discontinuously at $\phi=\phi_r+\pi$.
But, for finite $\langle N \rangle$, it is easily shown that $\langle\hat\Phi\rangle_\phi^{\phi_r}$ is a continuous function of $\phi$ (see Appendix \ref{sec:cont}).
Therefore it is not possible for the phase to be globally unbiased unless $\langle N \rangle$ is infinite.
Moreover, there must be a region of size scaling as $1/\langle N \rangle$ where the measurement is biased
[this follows from Eq.~\eqref{eq:cont}].

On the other hand, one can consider applying the CRB   to $\Delta_\phi \hat{\Phi}$ in Eq.~(\ref{eq:newv}); that is, to the {\rmse} modulo $(-\pi,\pi]$. 
Because $\Delta_\phi \hat{\Phi} \le \Delta_\phi^{0} \hat{\Phi}$  (see Sec.~II~A),  using the CRB to bound $\Delta_\phi \hat{\Phi}$ also yields a bound on $\Delta_\phi^{0} \hat{\Phi}$.
Also, because $\Delta_\phi \hat{\Phi} \equiv \Delta_\phi^{\phi} \hat{\Phi}$, the conditions for the measurement to be unbiased become
$b_\phi(\phi)=0$ and $b'_\phi(\phi)=0$.
It is important to note that $b'_\phi(\phi)$ is \emph{not} the same as $\frac{d}{d\phi} b_\phi(\phi)$.
In fact, the restriction $\frac{d}{d\phi} b_\phi(\phi)=0$ implies
\begin{align}
0 &= \left. \frac{d}{d\phi_r} b_{\phi_r}(\phi) \right|_{\phi_r = \phi} + b'_\phi(\phi) \nn
&= 2\pi p(\phi+\pi|\phi)+ b'_\phi(\phi).
\end{align}
That is, if $b_\phi(\phi)=0$, then $\frac{d}{d\phi} b_\phi(\phi)$ will automatically be zero, but $b'_\phi(\phi)$ will only be zero if $p(\phi+\pi|\phi)=0$.
In fact, for $b_\phi(\phi)=0$, the condition $b'_\phi(\phi)=0$ is equivalent to $p(\phi+\pi|\phi)=0$.

The conditions for the measurement to be unbiased (when applying the CRB to the {\rmse} modulo $(-\pi,\pi]$) can therefore be given as $b_\phi(\phi)=0$ and $p(\phi+\pi|\phi)=0$.
The condition $b_\phi(\phi)=0$ can be satisfied relatively easily, because it will be satisfied whenever the probability distribution for the error in the phase estimate is symmetric, so $p(\phi+\theta|\phi)=p(\phi-\theta|\phi)$.
However, it is not possible to satisfy $p(\phi+\pi|\phi)=0$ for all $\phi$ when $\langle N \rangle$ is small.
This is also the parameter regime where the HHB without bias must break down, because it would predict an impossibly large uncertainty. 

Hence, the bias of a given estimate is crucial in any application of the Cram\'{e}r-Rao bound to the RMSE. This is in strong contrast to the Heisenberg-type bounds for the RAMSE derived in this paper, which are independent of the bias function. 

\subsection{Asymptotic achievability}
It is often stated that the CRB (and QCRB and HHB) is asymptotically achievable in the limit of many probe states, without any further qualification.  However, for example, it is important to note from Eq.~(\ref{bias}) that, in the asymptotic limit $m\rightarrow\infty$, the RMSE does not approach zero for a biased estimate --- it is always bounded below by  $|b_{\phi_r}(\phi)|$.   

Furthermore,  Fisher's theorem that Eq.~(\ref{bias}) is itself asymptotically achievable, as $m\rightarrow\infty$,  does not hold in all cases of physical interest \cite{durkin}.
In particular, this theorem assumes that there is a unique maximally likely estimate \cite{Fisher}.  However, 
this is not the case for many states considered in quantum phase estimation, including the 
NOON states as per Eq.~(\ref{zhang}), which are the states that minimize the QCRB. The reason is of course that 
there is nothing to distinguish phase shifts modulo $2\pi/n$, regardless of the number of samples, 
unless the phase shift is in fact already known to this accuracy. That is, there are $n$ maximally likely estimates, 
so Fisher's theorem does not apply. 

In contrast, the above qualifications do not apply to the Heisenberg-type bounds for the RAMSE derived in this paper, which are independent of the bias of the estimate, and which are asymptotically achievable in the sense described in Sec.~VI.

\subsection{Example}
\label{sec:ex}
There are  obvious phase estimates that are not unbiased, where the {\rmse} obtained is qualitatively different
from what would be expected from the Cram\'er-Rao bound without correcting for bias.
Consider a simple measurement with a single photon in a MZI, in the state
\begin{equation}
\rho_\phi = \frac 12 (\ket{0}\bra{0}+\ket{1}\bra{1}) + \frac v2  (e^{i\phi}\ket{0}\bra{1}+e^{-i\phi}\ket{1}\bra{0}),
\end{equation}
with visibility $v<1$.
The photon-counting measurement at the output of the interferometer gives probabilities of measurement results
\begin{equation}
p(\pm|\phi) = (1\pm v\cos\phi) /2.
\end{equation}
For the $+$ measurement result, the optimal (least-square-error) estimate is $\hat{\phi}=0$, and for the $-$ measurement result the optimal estimate is $\hat{\phi}=\pi$.
With these estimates, the {\rmse} is given by
\begin{equation}
\label{eq:exvar}
\Delta_\phi \hat{\Phi} = \sqrt{\phi^2 (1 + v\cos\phi) /2 + (\pi-|\phi|)^2 (1- v\cos\phi) /2}.
\end{equation}
The absolute value of $\phi$ is taken above to take account of the fact that the difference should be determined modulo $2\pi$.
 For $v=1$, the error is zero at $\phi=0$ and $\phi=\pi$. 

In contrast, using the inverse square-root of the Fisher information (as for the CRB without correcting for bias) would give the lower bound
\begin{equation}
\Delta_\phi \hat{\Phi} \ge \frac {\sqrt{1-v^2\cos^2\phi}}{v|\sin\phi|}.
\end{equation}
In the limit $v\to 1$, the uncorrected CRB gives a result exactly equal to 1.
This is already greater than the actual {\rmse} for some $\phi$.
For imperfect visibility, the contrast is even stronger.
The uncorrected bound \emph{diverges} at $\phi=0$ and $\pi$, even though the actual measurement error is a \emph{minimum} there.
This result is illustrated in Fig.~\ref{fig:paradox}.
It is therefore clear that the CRB can give completely misleading results if it is not corrected for bias.
On the other hand, correcting the CRB for bias, via Eq.~(\ref{bias}), yields a bound exactly equal to the RMSE \eqref{eq:exvar}.

The QCRB in Eq.~(\ref{chain}), which assumes zero bias, gives $1/v$, and is also violated near $\phi = 0$ and $\phi = \pm \pi$
 by the biased measurements considered here.
Similarly, the HHB in Eq.~(\ref{chain}) yields a lower bound of $1$ (the same as the QCRB for $v=1$), which is also violated by the biased measurements considered here.
Thus we can see that caution needs to be employed in using Eq.~(\ref{chain}), because it requires unbiased measurements.
Such measurements are impossible in some cases, and even reasonable measurements can give highly biased estimates, resulting in a violation of the QCRB and HHB in Eq.~(\ref{chain}).

It is also interesting to compare these results to the error propagation formula, which is often used to estimate the measurement error.
The error propagation formula leads to the estimate of the error (using measurement operator $X=\ket{0}\bra{1}+\ket{1}\bra{0}$),
\begin{align}
\Delta_\phi \hat{\Phi} &\approx \frac{\sqrt{\langle X^2\rangle_\phi-\langle X\rangle_\phi^2}}{\left|\frac{d}{d\phi}\langle X\rangle_\phi\right|} \nn
&= \frac{\sqrt{1-v^2\cos^2\phi}}{v|\sin\phi|}.
\end{align}
Thus the error propagation formula gives an estimate of the uncertainty that is identical to the uncorrected Cram\'er-Rao bound.

\begin{figure}[!t]
\centering
\includegraphics[width=0.47\textwidth]{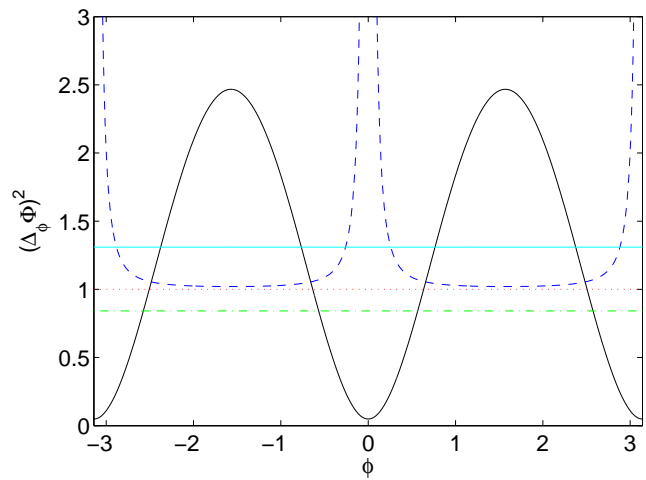}
\caption{{\mse} for phase measurements with a single photon using an interferometer with visibility $v=0.99$ and photon counting at the outputs. 
The actual {\mse}, as well as the Cram\'er-Rao bound with the correction for bias, is given by the solid curve (black).
The uncorrected Cram\'er-Rao bound, as well as the estimate given by the error propagation formula, is given as the dashed curve (dark blue).
The horizontal dash-dotted line  (green)  is the conjectured bound on the {\amse}, $k_C^2/\langle N+1 \rangle^2$.
The horizontal solid line  (light blue)  is the actual {\amse} for these measurements (obtained by averaging the {\mse} over $\phi$), 
and the horizontal dotted line  (red)  is the Helstrom-Holevo bound, which in this case is only slightly smaller than the quantum Cram\'er-Rao bound of $1/v$.}
\label{fig:paradox}
\end{figure}

\section{Conclusions}
We have rigorously proven that the square root of the average mean-square error ({\ramse}) of phase measurements is lower bounded 
by the Heisenberg limit $k/\langle G + 1\rangle$ [Eq.~\eqref{eq:heis0}]. 
The inequality with $k=k_A\approx 0.56$ holds in the case where the generator of the phase shifts has nonnegative integer eigenvalues.
We obtain a very similar result in the case where $G$ also has negative integer eigenvalues.
The result is as in Eq.~\eqref{eq:abheis}, where the absolute value of $G$ is used, and the scaling constant is again $k_A$.

These results mean that the accuracy of super-Heisenberg measurement schemes is essentially illusory.
They may work for a small range of phases, but if one considers the additional resources needed to locate an unknown phase to within the required range, the overall measurement will not violate the Heisenberg limit.

A new feature of our form of the Heisenberg limit is that it holds for all $\langle G \rangle$, not just in the asymptotic limit of large $\langle G \rangle$.
We achieve this by adding $1$ to the denominator.
This modification is necessary, because otherwise the inequality would indicate that the error must approach infinity in the limit $\langle G \rangle \to 0$.
This is impossible because phase has a bounded range. 

As well as the analytical result stated above, we have very  powerful evidence for a stronger  bound 
with $k_A$ replaced by $k_C\approx 1.38$.
We  have provided extensive  numerical evidence that the inequality holds with this larger scaling constant,
both in the case of the {\ramse} and for the square root of the Holevo variance.
In the case where $G$ has negative eigenvalues, the numerical results indicate that Eq.~\eqref{eq:abheis} holds with the scaling constant $k'_C\approx 0.79$, which is again larger than $k_A$.

 These stronger lower bounds are also supported by  asymptotic expansions of the exact solution for minimal Holevo variance, both for generators with nonnegative eigenvalues and generators without this restriction.
A similar result for the RAMSE has also been obtained, via an asymptotic expansion of a lower bound for this quantity, for the case of a generator that is \emph{not} restricted to nonnegative eigenvalues.
The case where the eigenvalues of $G$ are restricted to nonnegative eigenvalues is a possible area for future study.
The asymptotic expansions also enable us to show that these stronger lower bounds are asymptotically achievable.
That is, the minimum {\ramse} is equal to the lower bounds to leading order. 

We showed how various schemes that have been proposed to break the 
Heisenberg limit  do not break our bound on the {\ramse}.
The primary reason for this is that they can only violate the Heisenberg limit scaling if the phase shift is already known,
not when averaging over the phase shift.
Another factor is that they typically consider the Cram\'er-Rao bound, which is problematic for the Heisenberg limit.
It cannot provide a nontrivial lower bound for fixed mean photon number, as is required for the Heisenberg limit.
In addition, it requires knowledge of the bias.
Our alternative approach circumvents these limitations.

Our bound also differs from the Cram\'er-Rao bound in that it scales as $1/m$ in the number of copies of the state.
The Cram\'er-Rao bound scales as $1/\sqrt{m}$, but we have found that such scaling is impossible for a fixed $\langle N \rangle$.
This indicates that it is fundamentally impossible to obtain the Heisenberg limit from the Cram\'er-Rao bound.

\acknowledgments
DWB is funded by an ARC Future Fellowship (FT100100761).  MJWH, MZ, and HMW are supported by the ARC Centre of Excellence CE110001027. 

\appendix

\section{Details for the proof of Lemma 2}%\ref{lemz}
\label{sec:lem2det}
There are two results used in the proof of Lemma \ref{lemz}, which are proven here.
 First,  we show that for this operator the distribution of $G^{(s)}$ is the same as the distribution of $G$ for $\rho_0$.
The normalisation condition for the covariant measurement gives
\begin{equation}
\openone = \int d\theta \, e^{-iG\theta} \overline{M}_{0}e^{iG\theta},
\end{equation}
so
\begin{align}
\delta_{n,n'} \delta_{d,d'} &= \int d\theta \, \bra{n,d} e^{-iG\theta} \overline{M}_{0}e^{iG\theta} \ket{n',d'} \nn
&= \int d\theta \, e^{i(n'-n)\theta} \bra{n,d} \overline{M}_{0} \ket{n',d'} \nn
&= 2\pi \delta_{n,n'} \bra{n,d} \overline{M}_{0} \ket{n',d'}.
\end{align}
This means that $\langle n,d|\overline{M}_0|n,d'\rangle=\delta_{d,d'}/2\pi$.
Then, evaluating the distribution for $G^{(s)}$ gives
\begin{align} \label{eq:samep}
\bra{n} \rho_0^{(s)}|n\rangle &= 2\pi \sum_{d,d'\leq D(n)} \langle n,d'|\rho_0|n,d\rangle \, \langle n,d|\overline{M}_0|n,d'\rangle \nn
&= 2\pi \sum_{d,d'\leq D(n)} \langle n,d|\rho_0|n,d'\rangle \,(2\pi)^{-1}\delta_{d,d'} \nn
&= {\rm Tr}(\rho_0 P_n),
\end{align}
where $P_n:=\sum_d |n,d\rangle\langle n,d|$ denotes the projection onto eigenvalue $n$ of $G$.
The expression in the last line is the distribution of $G$ for $\rho_0$.

 Second,  we must show that $\rho_0^{(s)}$ is positive and has trace one, and is therefore a valid density operator.
Note one can always write the positive operators $\rho_0$ and $\overline{M}_0$ as sums of (not necessarily normalised or orthogonal) kets:
\begin{equation}
\rho_0 = \sum_\lambda |\lambda\rangle\langle\lambda|,~~\overline{M}_0 = \sum_\mu |\mu\rangle\langle\mu|.
\end{equation}
Hence, for any state $|\psi\rangle=\sum_n\psi_n|n\rangle$, 
\begin{align}
\langle \psi|\rho_0^{(s)}|\psi\rangle &= 2\pi \sum_{n,n'\in S,d\le D(n),d'\le D(n')} \psi_n\psi_{n'}^* \nn
& \quad \times \langle n',d'|\rho_0|n,d\rangle\langle n,d|\overline{M}_0|n',d'\rangle \nn
&= 2\pi \sum_{\lambda,\mu} | X_{\lambda,\mu} |^2 \geq 0 ,
\end{align}
where
\begin{equation}
X_{\lambda,\mu}:=\sum_{n\in S,d \le D(n)}\psi_n\langle\lambda|n,d\rangle\langle n,d|\mu\rangle.
\end{equation}
Hence $\rho_0^{(s)}\ge 0$, as required.
 Summing Eq.~\eqref{eq:samep} over $n$ yields ${\rm Tr}(\rho_0^{(s)})=1$, so $\rho_0^{(s)}$ is a valid density operator. 

\section{Details for Eq.~(45)} %\eqref{eq:nonint}
\label{detail}
Here we give the details of  the derivation  of Eq.~\eqref{eq:nonint}.
 First, variation of Eq.~(\ref{Lambda}) gives the optimising distribution 
\begin{equation}
p_n = e^{-(\alpha+1)} e^{-\beta|n-g|}, 
\end{equation}
which is the double exponential (Laplace) distribution found  in \cite{Hall12}. 
In Ref.~\cite{Hall12} it was assumed that $g$ was an integer, but  here we consider the more general case that $g$ is an arbitrary real number.

We require  $\beta>0$  in order for the distribution to be normalisable.
Then normalisation gives the restriction
\begin{equation}
 \frac{e^{-\beta r}+e^{-\beta(1-r)}}{1-e^{-\beta}} =  e^{\alpha+1},
\end{equation}
where $r=\lceil g \rceil - g$.
We then find that the mean value is
\begin{equation}
\langle |G-g| \rangle = \frac{(e^{\beta r}+e^{-\beta r})(1-r)+(e^{\beta(1-r)}+e^{-\beta (1-r)})r}{(1-e^{-\beta})(e^{\beta r}+e^{\beta(1-r)})},
\end{equation}
and
\begin{align}\label{halfway}
H(G) &= (\alpha+1) + \beta \langle|G-g| \rangle \nn
&= \ln \left( \frac{e^{-\beta r}+e^{-\beta(1-r)}}{1-e^{-\beta}}\right) +  \beta \langle|G-g| \rangle.
\end{align}

Without loss of generality we take $r\in[0,1/2]$.
The problem is symmetric about $r=1/2$, so these results also apply to $r>1/2$.
Then we obtain
\begin{align}
& 2\langle|G-g| \rangle \! + \! 1 - \exp \left[\beta r +\ln\left( \frac{e^{-\beta r}+e^{-\beta(1-r)}}{1-e^{-\beta}}\right) \right] \nn
&= \frac{(e^{\beta}-e^{2\beta r})[1+e^{2\beta r} + 2r(e^{\beta}-1)]}{(e^{\beta}-1)(e^{\beta}+e^{2\beta r})}\ge 0.
\end{align}
This then yields
\begin{equation}
 \beta r +\ln\left( \frac{e^{-\beta r}+e^{-\beta(1-r)}}{1-e^{-\beta}}\right) \le \ln (2\langle|G-g| \rangle+1).
\label{eq:fstpt}
\end{equation}

It can be shown that
\begin{align}
\label{eq:insim}
&e^\beta -(1+\beta) - e^{2\beta r} [e^{-\beta}(1+2\beta r)+\beta-1-2\beta r] \nn
&= (1+\beta r-\beta \langle|G-g| \rangle)(1-e^{-\beta})(e^{\beta}+e^{2\beta r}).
\end{align}
Next, for $\beta >  0$ and $r\ge 0$ we have
\begin{align}
0 & <  \beta (1-2r)^2+4r \nn
&= 1+\beta (1-2r) -[1+2(\beta -2)r-4\beta r^2] \nn
& \le e^{\beta (1-2r)} -[1+2(\beta-2)r-4\beta r^2].
\end{align}
This means that
\begin{align}
&\frac d{d\beta} \left\{ e^\beta -(1+\beta) - e^{2\beta r} [e^{-\beta}(1+2\beta r)+\beta-1-2\beta r] \right\} \nn
&= (1-e^{-\beta})e^{2\beta r} \left\{e^{\beta(1-2r)} -[1+2(\beta-2)r-4\beta r^2] \right\} \nn
&  > 0
\end{align}
 for $\beta>0$. 
This implies that
\begin{equation}
 e^\beta -(1+\beta) - e^{2\beta r} [e^{-\beta}(1+2\beta r)+\beta-1-2\beta r]  >  0,
\end{equation}
 for $\beta > 0$, 
because the left-hand side is zero for $\beta=0$,  and has positive  slope for  $\beta> 0$. 
Using Eq.~\eqref{eq:insim}, this gives
\begin{equation}
-\beta r + \beta \langle|G-g| \rangle  <  1.
\label{eq:secpt}
\end{equation}

Now adding Eqs.~\eqref{eq:fstpt} and \eqref{eq:secpt} yields
\begin{equation}
\ln\left( \frac{e^{-\beta r}+e^{-\beta(1-r)}}{1-e^{-\beta}}\right) +  \beta \langle|G-g| \rangle  <  \ln (2\langle|G-g| \rangle+1)+1,
\end{equation}
 and substitution into Eq.~(\ref{halfway}) gives Eq.~\eqref{eq:nonint} as required. 

 \section{Inequality proofs} 
\label{sec:tighter}
 To prove $\theta^2 \le f_3(\theta)$, consider the function
\begin{align}
\Delta(\theta) &= \sqrt{ (\pi^2/2)(1-\cos\theta) - (\pi^2/4-1)(1-\cos2\theta)/2} \nn& \quad - \theta.
\end{align}
Taking the derivative with respect to $\theta$ and solving to find the turning points of $\Delta(\theta)$ yields only two in the range $[0,\pi]$.
One is at $\theta=0$, and the other is at $\theta\approx 2.23$.
As $\Delta(\theta)>0$ for $\theta\approx 2.23$, and $\Delta(\theta)=0$ for $\theta=0$ or $\pi$, we have $\Delta(\theta)\ge 0$ for $\theta\in[0,\pi]$.
This proves Eq.~\eqref{eq:inva} for $\theta\in[0,\pi]$, and the result for $\theta\in[-\pi,0]$ follows because Eq.~\eqref{eq:inva} is symmetric.

Next, to prove a lower bound on $\delta\Phi$, we use Eq.~(\ref{earlier}), which was 
\begin{equation}
\langle \cos\Theta \rangle \ge \cos\sqrt{\langle \Theta^2\rangle} .
\end{equation}
Using this, we have 
\begin{equation}
\label{eq:lobnd}
\delta\hat\Phi = \sqrt{\langle \Theta^2\rangle} \ge \arccos\langle \cos\Theta \rangle = \arccos[1-(\delta_1\hat\Phi)^2/2].
\end{equation}
Now note that, if we have a state that minimises $\delta\hat\Phi$, then it can not give a value of $\arccos(1-(\delta_1\hat\Phi)^2/2)$ smaller than that for the minimum value of $\delta_1\hat\Phi$.
This means that we can lower bound $\delta\hat\Phi$ by the minimum value of $\arccos[1-(\delta_1\hat\Phi)^2/2]$.
This is a tighter lower bound on $\delta\hat\Phi$ than $\delta_1\hat\Phi$, because $\arccos(1-x^2/2)=x+x^3/24+O(x^5)$.
Unfortunately $\arccos[1-(\delta_1\hat\Phi)^2/2]$ can still be below $k_C/\langle N+1 \rangle$.

\section{Asymptotic behaviour for Holevo variance}
\label{asymp}
 Here we derive the asymptotic results for the variance given in Sec.~\ref{sec:asymp}.
Using Eq.~(12) of Ref.~\cite{hansen} (with $q=p+1$), one has
\begin{align}
\sum_{k=1}^\infty J_{x+k}(z)\,J_{x+k+1}(z) &= \frac{z}{2}\left[J_{x+1}^2(z) - J_x(z)\,J_{x+2}(z)\right] \nn
&= \frac{z}{2} [J_{x+1}(z)]^2 ,
\end{align}
where the second equality follows from Eq.~(\ref{zero}).

Second, from Eq.~(32) of Ref.~\cite{hansen}, one has
\begin{align}
 \sum_{k=1}^\infty [J_{x+k}(z)]^2 &= \frac{z}{2}\left[ J_{x+1}(z)\frac{\partial J_x(z)}{\partial x}-J_x(z)\frac{\partial J_{x+1}(z)}{\partial x}\right] \nn
&= \frac{z}{2} J_{x+1}(z)\frac{\partial J_x(z)}{\partial x} ,
\end{align}
where the second equality similarly follows via Eq.~(\ref{zero}).
We therefore have
 \begin{equation}
 \langle e^{i\Theta}\rangle = \frac{J_{x+1}(z)}{[\partial J_x(z)/\partial x]} .
 \end{equation}
Using Eq.\ \eqref{eigsol} then yields
\begin{equation}
\label{eq:nav}
\langle N+1 \rangle = \frac{z J_{x+1}(z)}{[{\partial J_x(z)/}{\partial x}]}-x.
\end{equation}

There are a number of asymptotic results that we can use.  From \cite{Elbert,Oliver}
\begin{align}
\label{zeq}
z &= x + \gamma x^{1/3} + \frac{3\gamma^2}{10 x^{1/3}} + \frac{5-\gamma^3}{350 x} - \frac{479\gamma^4+20\gamma}{63000 x^{5/3}} \nn
&\quad +\frac{20231\gamma^5-27550\gamma^2}{8085000 x^{7/3}} + O(x^{-3}),
\end{align}
where $\gamma=|z_A|/2^{1/3}$,  and  $z_A$ is the first zero of the Airy function.  We can invert this relation to give
\begin{align}
\label{xeq}
x &= z - \gamma z^{1/3} + \frac{\gamma^2}{30 z^{1/3}} - \frac{5-\gamma^3}{350 z} + \frac{281 \gamma^4-5220\gamma}{567000 z^{5/3}} \nn
& \quad + \frac{73769\gamma^5-3312450\gamma^2}{654885000 z^{7/3}}+ O(z^{-3}).
\end{align}
Now from Eq.\ (9.3.23) of \cite{abrams} we have, for $z=x+yx^{1/3}$,
\begin{align}
\label{abram}
J_x(x+yx^{1/3}) &\sim \frac{2^{1/3}}{x^{1/3}} {\rm Ai} (-2^{1/3} y) \left( 1+\sum_{k=1}^\infty \frac{f_k(y)}{x^{2k/3}}\right) \nn
& \quad + \frac{2^{2/3}}{x} {\rm Ai}' (-2^{1/3} y) \sum_{k=0}^\infty \frac{g_k(y)}{x^{2k/3}}.
\end{align}
The functions are given by
\begin{align}
f_1(y) &= -\frac 15 y, \\
f_2(y) &= -\frac 9{100}y^5+\frac 3{35}y^2, \\
f_3(y) &= \frac{957}{7000}y^6-\frac{173}{3150}y^3-\frac 1{225}, \\
f_4(y) &= \frac{27}{20000}y^{10} - \frac{23573}{147000} y^7 + \frac{5903}{138600} y^4 + \frac{947}{346500}y, \\
g_0(y) &= \frac 3{10} y^2, \\
g_1(y) &= -\frac {17}{70}y^3+\frac 1{70}, \\
g_2(y) &= -\frac{9}{1000}y^7+\frac{611}{3150}y^4-\frac {37}{3150}y, \\
g_3(y) &= \frac{549}{28000}y^8 - \frac{110767}{693000} y^5 + \frac{79}{12375} y^2.
\end{align}

Using Eq.\ \eqref{zeq} we have
\begin{align}
y&=(z-x)/x^{1/3} \nonumber \\
&=  \gamma + \frac{3\gamma^2}{10 x^{2/3}} + \frac{5-\gamma^3}{350 x^{4/3}} - \frac{479\gamma^4+20\gamma}{63000 x^2} \nn
& \quad +\frac{20231\gamma^5-27550\gamma^2}{8085000 x^{8/3}} + O(x^{-3}).
\end{align}
We can then substitute Eq.\ \eqref{xeq}, which gives
\begin{align}
y - \gamma &= \frac{3\gamma^2}{10 z^{2/3}} +  \frac{5+69\gamma^3}{350 z^{4/3}} + \frac{9361\gamma^4+1180\gamma}{63000 z^2} \nn
&\quad +\frac{8691349\gamma^5+1484550\gamma^2}{72765000 z^{8/3}} + O(z^{-3}).
\end{align}

Now to take the derivative with respect to the order, we can use
\begin{align}
\label{eq:derv}
\frac{\partial}{\partial x} J_x(z) &= \frac d{dx} J_x(x+yx^{1/3}) + \frac{dy}{dx}\frac d{dy} J_x(x+yx^{1/3}) \nonumber \\
&= \frac d{dx} J_x(x+yx^{1/3}) \nn
& \quad - \left(\frac z{3x^{4/3}}+\frac 2{3x^{1/3}}\right)\frac d{dy} J_x(x+yx^{1/3}) .
\end{align}
In the resulting expression it is possible to expand in a series for $y$ about $\gamma$, then expand in a series about $z$.

It is possible to determine a series in $z$ for
\begin{equation}
\label{eq:nser}
 \frac{z J_{x+1}(z)}{[{\partial J_x(z)/}{\partial x}]}-x.
\end{equation}
We can then invert this series, finding a series in  $\langle N+1 \rangle$  for $z$.
Similarly, it is possible to find a series in $z$ for
\begin{equation}
\frac{J_{x+1}(z)}{[{\partial J_x(z)/}{\partial x}]}.
\end{equation}
Then we can express
\begin{equation}
\left(\frac{J_{x+1}(z)}{[{\partial J_x(z)/}{\partial x}]}\right)^{-2}-1,
\end{equation}
as a series in $z$.
Substituting the series for $z$ in $\langle N +1 \rangle$, the overall result is as in Eq.~\eqref{eq:series}, with
\begin{align}
b_2 &= \frac{4|z_A|^3}{27} = k_C^2 \approx 1.8936, \\ 
b_4 &= \frac{16|z_A|^6}{1215} \approx 2.1514, \\ 
b_6 &= \frac{16|z_A|^6(27+40|z_A|^3)}{688905} \approx 2.0424, \\
b_8 &= \frac{256|z_A|^9(3+|z_A|^3)}{4428675} \approx  1.9050, \\ 
b_{10} &= \frac{64|z_A|^9(2673+9252|z_A|^3+1120|z_A|^6)}{21483502425} \approx 1.8906.
\end{align}

Next, to place an upper bound on the {\ramse} for this state, we use Eq.~\eqref{eq:inva}.
 This equation gives 
\begin{equation}
\langle \Theta^2 \rangle \le (\pi^2/2)(1-\langle\cos\Theta\rangle) - (\pi^2/4-1)(1-\langle\cos2\Theta\rangle)/2.
\end{equation}
To find the expectation value $\langle e^{i2\Theta}\rangle$, we use
\begin{align}
\langle e^{i2\Theta}\rangle &= A^2 \sum_{n=0}^\infty J_{x+n+1}(z)\,J_{x+n+3}(z) \nn
&= \frac{\sum_{k=1}^\infty J_{x+k}(z)\,J_{x+k+2}(z) }{\sum_{k=1}^\infty [J_{x+k}(z)]^2}.
\end{align}
Again using Eq.~(12) of Ref.~\cite{hansen}, but now with $q=p+2$, we have
\begin{align}
&\sum_{k=1}^\infty J_{x+k}(z)\,J_{x+k+2}(z) \nn
&= \frac{z}{4}\left[J_{x+1}(z)J_{x+2}(z) - J_x(z)\,J_{x+3}(z)\right] \nn
&= \frac{z}{4} J_{x+1}(z)J_{x+2}(z) .
\end{align}
That then gives
\begin{equation}
\langle e^{i2\Theta}\rangle = \frac{J_{x+2}(z)}{2[\partial J_x(z)/\partial x]} .
\end{equation}

Using this expression, and expanding the Bessel function solution in a series as above, we obtain
\begin{align}
\langle \Theta^2 \rangle &\le (\pi^2/2)(1-\langle\cos\Theta\rangle) - (\pi^2/4-1)(1-\langle\cos2\Theta\rangle)/2 \nn
&= \frac{4|z_A|^3}{27\langle N+1 \rangle^2} + \frac{(\pi^2-4)|z_A|^3}{54\langle N+1 \rangle^3} + O\left(\frac 1{\langle N+1 \rangle^4}\right).
\end{align}
This results in the inequality given in Eq.~\eqref{eq:upbnd}.

Expanding in a series also gives
\begin{align}
(\arccos\langle \cos\Theta \rangle)^2 &= \frac{4|z_A|^3}{27\langle N+1 \rangle^2} - \frac{16|z_A|^6}{10935\langle N+1 \rangle^4} \nn
&\quad + O\left(\frac 1{\langle N+1 \rangle^6}\right).
\end{align}
Using Eq.~\eqref{eq:lobnd}, we therefore have the upper and lower bounds on $(\delta\hat{\Phi})^2$,
\begin{align}
&\frac {k_C^2}{\langle N+1 \rangle^2}- O\left(\frac 1{\langle N+1 \rangle^4}\right) \le (\delta\hat{\Phi})^2 \nn
&\le \frac {k_C^2}{\langle N+1 \rangle^2}+ O\left(\frac 1{\langle N+1 \rangle^3}\right).
\end{align}

\section{Asymptotic behaviour for variance with fixed $\langle|J|\rangle$}
\label{asymp2}
The normalization constraint yields, using Eq.~(32) of Ref.~\cite{hansen},
\begin{align}
A^{-2} &= [J_{x}(z)]^2 + 2\sum_{j=1}^\infty [J_{x+j}(z)]^2 \nonumber \\
&= [J_{x}(z)]^2 + z\left[ J_{x+1}(z)\frac{\partial J_x(z)}{\partial x}-J_x(z)\frac{\partial J_{x+1}(z)}{\partial x}\right].
\end{align}
We also have
\begin{align} \nonumber
\langle |J|\rangle &= 2 A^2 \sum_{j=0}^\infty j\,[J_{x+j}(z)]^2 , \\
\langle e^{i\Theta}\rangle &= A^2 \left[ 2J_{x}(z)\,J_{x+1}(z)+2\sum_{j=1}^\infty J_{x+j}(z)\,J_{x+j+1}(z) \right].
\end{align}
Using Eq.~(12) of \cite{hansen} we have
\begin{align}
\langle e^{i\Theta}\rangle &= A^2 \left[ 2J_{x}(z)\,J_{x+1}(z)+zJ_{x+1}^2(z) \right. \nonumber \\
&\quad  \left. - zJ_x(z)\,J_{x+2}(z)\right] .
\end{align}
In this case we have
\begin{equation}
\label{eigsol2}
\langle e^{i\Theta}\rangle = (\alpha+\beta\langle |J| \rangle) = (x+\langle |J| \rangle )/z,
\end{equation}
so
\begin{equation}
 \langle |J| \rangle = z\langle e^{i\Theta}\rangle -x.
\end{equation}

The first zero of the derivative of the Bessel function is given by \cite{Elbert,Oliver}
\begin{align}
z&=x+\gamma' x^{1/3} + \left( \frac {3\gamma'^2}{10}-\frac 1{10\gamma'}\right)\frac 1{x^{1/3}} \nonumber \\
&\quad  - \left( \frac {\gamma'^3}{350}+\frac 1{25}+\frac 1{200\gamma'^3}\right)\frac 1x \nonumber \\
& \quad - \frac{958\gamma'^9-2036\gamma'^6-84\gamma'^3+63}{126000\gamma'^5 x^{5/3}} + O(x^{-7/3}).
\end{align}
We have corrected an error in Ref.~\cite{Elbert} where ``840'' was given instead of ``84''.
Here $\gamma'=|z'_A|/2^{1/3}$, where $z'_A$ is the first zero of the derivative of the Airy function.
Inverting this series, and performing series expansions similar to that for the first case, gives
the series in Eq.~\eqref{eq:spinser} with
\begin{align}
d_2 &= \frac{16|z'_A|^3}{27} = {k'_C}^2 \approx  0.6266, \\
d_3 &= \frac{32|z'_A|^3}{27} \approx 1.2533, \\
d_4 &= \frac{16|z'_A|^3 (111 - 4 |z'_A|^3)}{1215} \approx 1.4868, \\
d_5 &= \frac{64 |z'_A|^3 (21 - 4 |z'_A|^3)}{1215} \approx 0.9341, \\
d_6 &= \frac{16 |z'_A|^3 (-63 - 40488 |z'_A|^3 + 160 |z'_A|^6)}{1148175} \approx -0.6292.
\end{align}

 To determine an upper bound, we need to determine 
\begin{align}
\langle e^{i2\Theta}\rangle &= A^2 \left[ J_{x+1}^2(z)+ 2J_{x}(z)\,J_{x+2}(z) \right. \nn & \quad
\left.+2\sum_{j=1}^\infty J_{x+j}(z)\,J_{x+j+2}(z) \right] \nn
&=  A^2 \left\{ J_{x+1}^2(z)+ 2J_{x}(z)\,J_{x+2}(z) \right. \nn & \quad
\left.+\frac z2 [J_{x+1}(z)J_{x+2}(z) - J_x(z)\,J_{x+3}(z)] \right\}.
\end{align}
Expanding in a series then gives
\begin{align}
\langle \Theta^2 \rangle &\le \frac{16|z'_A|^3}{27\langle 2|J|+1\rangle^2} + \frac{32|z'_A|^3}{27\langle 2|J|+1\rangle^3} \nn
& \quad + O\left(\frac 1{\langle 2|J|+1 \rangle^4}\right).
\end{align}
Therefore, the {\mse}  is upper and lower bounded as
\begin{align}
&\frac {{k'_C}^2}{\langle 2|J|+1 \rangle^2}+ O\left(\frac 1{\langle 2|J|+1 \rangle^3}\right) \le (\delta\hat{\Phi})^2 \nn
&\le \frac {{k'_C}^2}{\langle 2|J|+1 \rangle^2}+ O\left(\frac 1{\langle 2|J|+1 \rangle^3}\right).
\end{align} 

\section{Continuity of the expected phase estimate}
\label{sec:cont}
Here we show that $\langle\hat\Phi\rangle_\phi^{\phi_r}$ is a continuous function of $\phi$.
Defining
\beq
X_{\phi_r} := \int_{\phi_r-\pi}^{\phi_r+\pi} d\hat\phi \, \hat\phi \, M_{\hat\phi},
\eeq
we have
\begin{equation}
\langle\hat\Phi\rangle_\phi^{\phi_r} = {\rm Tr} ( X_{\phi_r} \rho_\phi) .
\end{equation}
The expectation value of the phase estimate at $\phi+\epsilon$ is
\beq
\langle\hat\Phi\rangle_{\phi+\epsilon}^{\phi_r} = {\rm Tr}( X_{\phi_r} e^{-iG\epsilon} \rho_\phi e^{iG\epsilon} ).
\eeq
The difference is
\beq
|\langle\hat\Phi\rangle_{\phi+\epsilon}^{\phi_r} - \langle\hat\Phi\rangle_\phi^{\phi_r}| = |{\rm Tr} [X_{\phi_r} (e^{-iG\epsilon} \rho_\phi e^{iG\epsilon} - \rho_{\phi})]|.
\eeq

Take $\ket{\xi_j}$ to be the eigenbasis of $X_{\phi_r}$.
Then
\begin{align}
|\bra{\xi_j} X_{\phi_r} \ket{\xi_j}|  &\le  \int_{\phi_r-\pi}^{\phi_r+\pi} d\hat\phi \, |\hat\phi| \, \bra{\xi_j} M_{\hat\phi} \ket{\xi_j} \nn
&\le  2\pi \int_{\phi_r-\pi}^{\phi_r+\pi} d\hat\phi \, \bra{\xi_j} M_{\hat\phi} \ket{\xi_j} \nn
&=  2\pi \bra{\xi_j} \openone \ket{\xi_j} =2\pi.
\end{align}
Using this, we find
\begin{align}
& |{\rm Tr} [X_{\phi_r} (e^{-iG\epsilon} \rho_\phi e^{iG\epsilon} - \rho_\phi)]| \nn
&= \left| \sum_{j} \bra{\xi_j} X_{\phi_r} \ket{\xi_j}\bra{\xi_j} (e^{-iG\epsilon} \rho_\phi e^{iG\epsilon} - \rho_\phi)\ket{\xi_j}  \right| \nn
&\le 2\pi \sum_{j}  \left| \bra{\xi_j} (e^{-iG\epsilon} \rho_\phi e^{iG\epsilon} - \rho_\phi)\ket{\xi_j}  \right|.
\end{align}
Take $\ket{\zeta_j}$ to be the eigenbasis of $e^{-iG\epsilon} \rho_\phi e^{iG\epsilon} - \rho_\phi$.
Then
\begin{align}
& \sum_{j}  \left| \bra{\xi_j} (e^{-iG\epsilon} \rho_\phi e^{iG\epsilon} - \rho_\phi)\ket{\xi_j}  \right| \nn
&= \sum_{j}  \left| \sum_{k} |\braket{\xi_j}{\zeta_k}|^2 \bra{\zeta_k} (e^{-iG\epsilon} \rho_\phi e^{iG\epsilon} - \rho_\phi)\ket{\zeta_k}  \right| \nn
&\le \sum_{k} \left|\bra{\zeta_k} (e^{-iG\epsilon} \rho_\phi e^{iG\epsilon} - \rho_\phi)\ket{\zeta_k}  \right| \nn
&=\| e^{-iG\epsilon} \rho_\phi e^{iG\epsilon} - \rho_\phi \|_1.
\end{align}
Hence
\beq
|{\rm Tr} [X_{\phi_r} (e^{-iG\epsilon} \rho e^{iG\epsilon} - \rho_\phi)]| \le 2\pi \| e^{-iG\epsilon} \rho_\phi e^{iG\epsilon} - \rho_\phi \|_1 \, .
\eeq

Take the state to be given by
\begin{equation}
\rho_\phi = \sum_j p_j \ket{\psi_j}\bra{\psi_j}.
\end{equation}
For $\ket{\psi_j}$,
\begin{align}
\bra{\psi_j} e^{-iG\epsilon} \ket{\psi_j} &= \sum_k |\psi_{jk}|^2 e^{-ik\epsilon} \nn
&= 1-\sum_k |\psi_{jk}|^2 (1-e^{-ik\epsilon}).
\end{align}
Evaluating the distance from 1 gives
\begin{align}
\left|\sum_k |\psi_{jk}|^2 (1-e^{ik\epsilon})\right| &\le \sum_k |\psi_{jk}|^2 |1-e^{-ik\epsilon}| \nn
&\le \sum_k |\psi_{jk}|^2 |k\epsilon| \nn
&= \bra{\psi_j} |G| \ket{\psi_j}|\epsilon|,
\end{align}
so
\begin{align}
D(\ket{\psi_j}, e^{-iG\epsilon} \ket{\psi_j}) &= 2\sqrt{1-|\bra{\psi_j} e^{-iG\epsilon} \ket{\psi_j}|^2} \nn
&\le 2\sqrt{1-|1-\bra{\psi_j} |G| \ket{\psi_j} |\epsilon||^2} \nn
&\le 2\sqrt{2\bra{\psi_j} |G| \ket{\psi_j} |\epsilon|}.
\end{align}
By the convexity of trace distance
\begin{align}
\| e^{-iG\epsilon} \rho_\phi e^{iG\epsilon} - \rho_\phi \|_1 & \le \sum_j p_j D(\ket{\psi_j}, e^{-iG\epsilon} \ket{\psi_j}) \nn
&\le \sum_i p_i 2\sqrt{2\bra{\psi_j} |G| \ket{\psi_j} |\epsilon|} \nn
&\le 2\sqrt{2\langle |G|\rangle |\epsilon|},
\end{align}
where $\langle |G| \rangle={\rm Tr}(|G|\rho)$. Hence
\beq
\label{eq:cont}
|\langle\hat\Phi\rangle_{\phi+\epsilon}^{\phi_r} - \langle\hat\Phi\rangle_\phi^{\phi_r}| \le 4\pi \sqrt{2\langle |G|\rangle |\epsilon|}.
\eeq
Thus the expectation value of the phase estimate must be a continuous function of $\phi$ unless $\langle |G|\rangle$ is infinite.

\end{document}